\documentclass[pra,aps,floatfix,amsmath,superscriptaddress,twocolumn]{revtex4}
\usepackage{amssymb,enumerate}
\usepackage{graphicx}
\usepackage{graphics}
\usepackage{amsmath}
\usepackage{amsthm,bbm}
\usepackage{color}
\usepackage{dsfont}
\usepackage{hyperref}

\usepackage{mathtools}

\def\>{\rangle}
\def\<{\langle}

\def\id{\mathsf{id}}
\def\mE{\mathcal{E}}
\def\mB{\mathcal{B}}
\def\mF{\mathcal{F}}

\def\mG{\mathcal{G}}
\def\sH{\mathcal{H}}

\renewcommand{\qedsymbol}{\nobreak \ifvmode \relax \else
	\ifdim \lastskip<1.5em \hskip-\lastskip \hskip1.5em plus0em
	minus0.5em \fi \nobreak \vrule height0.75em width0.5em
	depth0.25em\fi}

\renewcommand{\ge}{\geqslant}
\renewcommand{\le}{\leqslant}
\renewcommand{\geq}{\geqslant}
\renewcommand{\leq}{\leqslant}

\newtheorem{theorem}{Theorem}
\newtheorem{corollary}{Corollary}
\newtheorem{lemma}{Lemma}

\newtheorem{definition}{Definition}
\newtheorem*{lemma*}{Lemma}

\theoremstyle{remark}
\newtheorem{remark}{Remark}

\theoremstyle{definition}

\newcommand{\bea}{\begin{eqnarray}}
\newcommand{\eea}{\end{eqnarray}}
\newcommand{\be}{\begin{equation}}
\newcommand{\ee}{\end{equation}}
\newcommand{\ba}{\begin{equation}\begin{aligned}}
\newcommand{\ea}{\end{aligned}\end{equation}}

\def\be{\begin{equation}}
\def\ee{\end{equation}}

\newcommand{\mU}{\mathcal{U}}

\newcommand{\mH}{\mathcal{H}}

\newcommand{\lr}{\rangle\langle}
\newcommand{\la}{\langle}
\newcommand{\ra}{\rangle}
\newcommand{\tr}{{\rm Tr}}

\newcommand{\mbf}[1]{\mathbf{#1}}
\newcommand{\mbb}[1]{\mathbb{#1}}

\def\>{\rangle}
\def\<{\langle}

\def\E{ {\mathcal E} }

\def\H{ {\mathcal H} }

\def\U {{\mathcal U}}

\def\G {{\mathcal G}}

\def\F{ {\mathcal F} }

\def\B{ {\mathcal B} }

\def\D{ {\mathcal D} }
\def\I{ \mathbbm{1} }

\def\r{\boldsymbol{r}}
\def\s{\boldsymbol{s}}
\def\p{\mathbf{p}}
\def\q{\mathbf{q}}
\def\e{\mathbf{e}}

\def\t{\boldsymbol{t}}

\newcommand{\Trr}[1]{\operatorname{Tr}\!\left[#1\right]}

\usepackage{soul,color}

\begin{document}
	
	
	\title{Quantum majorization and a complete set of entropic conditions\\for quantum thermodynamics}
	
	\author{Gilad \surname{Gour}}\email{gour@ucalgary.ca}
\affiliation{
Department of Mathematics and Statistics,
University of Calgary, AB, Canada T2N 1N4} 
\affiliation{
Institute for Quantum Science and Technology,
University of Calgary, AB, Canada T2N 1N4} 

\author{David \surname{Jennings}}\email{david.b.jennings@gmail.com}
\affiliation{Department of Physics, University of Oxford, Oxford, OX1 3PU, United Kingdom}
\affiliation{Department of Physics, Imperial College London, London SW7 2AZ, United Kingdom}
\author{Francesco \surname{Buscemi}}
\affiliation{Department of Computer Science and Mathematical Informatics, Nagoya University, Chikusa-ku, Nagoya,
	464-8601, Japan}
\email{buscemi@is.nagoya-u.ac.jp}

\author{Runyao \surname{Duan}}
\affiliation{Centre  for  Quantum  Software  and  Information,
Faculty  of  Engineering  and  Information  Technology,
University of Technology Sydney,  NSW 2007,  Australia}
\email{runyao.duan@uts.edu.au}

\author{Iman  \surname{Marvian}}
\affiliation{MIT Research Laboratory of Electronics, USA}
\email{imarvian@gmail.com}

	\date{\today}

	\begin{abstract}
 What does it mean for one quantum process to be more disordered than another? Here we provide a precise answer to this question in terms of a quantum-mechanical generalization of majorization. The framework admits a complete description in terms of single-shot entropies, and provides a range of significant applications. These include applications to the comparison of quantum statistical models and quantum channels, to the resource theory of asymmetry, and to quantum thermodynamics. In particular, within quantum thermodynamics, we apply our results to provide the first complete set of necessary and sufficient conditions for arbitrary quantum state transformations under thermodynamic processes, and which rigorously accounts for quantum-mechanical properties, such as coherence. Our framework of generalized thermal processes extends thermal operations, and is based on natural physical principles, namely, energy conservation, the existence of equilibrium states, and the requirement that quantum coherence be accounted for thermodynamically. In the zero coherence case we recover thermo-majorization   while in the asymptotic coherence regime we obtain a constraint that takes the form of a Page-Wootters clock condition.  
		\end{abstract}

	\maketitle
\section{Introduction}	

Irreversibility -- the loss of order and the increase of disorder --  is a fundamental and ubiquitous feature of physics that is typically described through thermodynamics and thermodynamic entropy. However, its scope goes above and beyond what one would ordinarily consider thermodynamic in nature. For example, the use of quantum entanglement within photonic quantum computing is subject to a form of irreversibility that need not be attached to either a particular energy scale or an equilibrium environment. Increasingly, a broader notion of irreversibility has been developed, that has been shown to include thermodynamic irreversibility as a special case, and has also allowed us to study intrinsically quantum mechanical order (such as entanglement or coherence) in contrast to classically ordered systems. Majorization is at the core of this development.

Majorization is a fundamental tool that finds application across a wide range of subjects from economics and statistics, to physics, chemistry and pure mathematics~\cite{Marshall-Olkin,here-there}. At its core lies a notion of ``deviations from uniformity", and the theory ties together mathematical techniques in convexity, combinatorics and partial orders.  

An example of its use is in statistical mechanics of a physical system with $N$ energy levels. If we assume, for the sake of discussion, that the system is fully degenerate in energy, its thermal equilibrium state is described by the probability distribution $\boldsymbol{\gamma}=(\frac{1}{N}, \dots, \frac{1}{N})$ over the energy levels. Given any two other probability distributions $\p = (p_1, \dots, p_N)$ and $\q=(q_1, \dots ,q_N)$ one might wish to say if one is more or less out of equilibrium than the other. Majorization provides a concrete way of stating this. The distribution $\p$ is more ordered than $\q$ (or ``$\p$ majorizes $\q$", written $\q \prec \p$) if $\q= D \p$ for some doubly stochastic matrix $D$ \cite{Marshall-Olkin}. A crucial property of majorization is that it can be equivalently formulated in terms of a \textit{complete set of monotones}. For example, it is well-known that $\q \prec \p$ if and only if $\sum_k f(p_k) \ge \sum_k f(q_k)$ for all continuous real-valued convex functions $f$. Together these imply that the value of any continuous convex function $f$ on statistical distributions can never increase under doubly stochastic transformations. Such functions are therefore \textit{monotones} that quantify the deviation from equilibrium; moreover, they constitute a \textit{complete} set of monotones because the comparison of their values provides a sufficient condition for the existence of a doubly stochastic transformation.

Majorization also finds extensive use in various parts of quantum information theory, such as in entanglement theory \cite{Nielsen99} and recent formulations of resource theories (See e.g. \cite{gour-mueller}). In particular, it plays a central role in the recent thermodynamic frameworks using quantum information theory \cite{janzing2000thermodynamic, horodecki2013fundamental, aberg2013truly, brandao2011resource,brandao2013second, gour-mueller, narasimhachar,lostaglio2015description, cwiklinski2014limitations,lostaglio2014quantum,lostaglioStochastic, weilenmann, korzekwa2015extraction,chiribella2015, mueller2017}. In particular it was found that state transformations with zero coherences in energy are fully characterised by thermo-majorization \cite{horodecki2013fundamental} (see also earlier works \cite{ruch,Dahl1999}), which is a natural generalisation of majorization~\cite{janzing2000thermodynamic,HO13,Ren16}. However it was shown in \cite{lostaglio2015description} that such thermo-majorization results are insufficient for describing quantum coherence under thermal operations, and that novel coherence measures are required. Low temperature coherence regimes were shown to admit solvable analysis \cite{narasimhachar}, general coherence bounds were developed \cite{cwiklinski2014limitations}, and a framework for coherence based on the concept of asymmetry under time-translations was proposed \cite{lostaglio2015description, lostaglio2014quantum}. However a complete specification of the structure of non-equilibrium quantum states was still lacking.

A natural question is therefore whether there exists a generalization of majorization (or thermo-majorization) that can accommodate such intrinsically quantum-mechanical orderings. Several candidate generalizations exist~\cite{Marshall-Olkin,FBarxiv,BG16}, however the one most relevant to our present work is called \emph{matrix majorization}~\cite{Dahl1999}, which is a specialization to linear algebra of ideas coming from the theory of statistical comparison (see Ref.~\cite{torgersen} and references therein). Given two matrices of real numbers $A$ and $B$, we say that $A$ matrix majorizes $B$, and write $B \prec_m A$, if and only if $B = AX$ for some row stochastic matrix $X$. It is easy to see that this is a generalization of majorization: for the two-row matrices $A=\begin{bmatrix}\p \\ \e\end{bmatrix}$ and $B=\begin{bmatrix}\q \\ \e\end{bmatrix}$ (with $\e\equiv(1,1,...,1)$) the relation $B \prec_m A$ is equivalent to 
$\q\prec\p$. Similarly, other variants of majorization, like thermo-majorization, are special cases of matrix majorization.  However, such an ordering is inherently classical, being ultimately based on stochasticity, as opposed to coherent quantum processes. A central component of our present work is to generalize matrix majorization in a natural way into the quantum-mechanical setting, and to provide applications to a number of topics, in particular a fully quantum-mechanical framework for thermodynamics.

\section{Definition of Quantum Majorization}

Our generalization of matrix majorization, which we call \emph{quantum majorization}, defines a relation on bipartite quantum states, and consequently, due to the channel-state duality property of quantum theory (i.e., the Choi isomorphism~\cite{Choi}), also defines a relation on quantum processes, i.e., completely positive and trace-preserving (CPTP) maps.

\begin{definition}\label{maindef}
Let $\rho^{AB}\in\B(\H_A\otimes\H_B)$ and $\sigma^{AC}\in\B(\H_A\otimes\H_C)$ be two quantum bipartite states.
We say that $\rho^{AB}$ \emph{quantum majorizes} $\sigma^{AC}$, and write $\sigma^{AC} \prec_q \rho^{AB}$, if there exists a CPTP map $\E:\B(\H_B) \to \B(\H_C)$ such that $ \id \otimes \E (\rho^{AB}) = \sigma^{AC}$. 
\end{definition}
\begin{remark}
The preorder $\sigma^{AC} \prec_q \rho^{AB}$ is not symmetric with respect to the action of $\mE$. It means that $\rho^{AB}$ quantum majorizes $\sigma^{AC}$ \emph{on $B$}. However, in the remaining of this paper, it will be clear from the text that the action of $\mE$ is on system $B$.
\end{remark}

It is clear from Definition~\ref{maindef} that $\rho^A=\sigma^A$ is a necessary condition, called the \emph{compatibility condition}, for the ordering of states to hold since $\E$ is trace-preserving,  and when it holds the two states are said to be \emph{compatible}. Moreover, in the special case that the marginals satisfy $\rho^A=\sigma^A=\frac{1}{d_A}\I^A$, we can express the bipartite states as the Choi matrices $\rho^{AB}=\id\otimes\D\left(\phi_{+}^{AA'}\right)$ and $\sigma^{AC}=\id\otimes\F\left(\phi_{+}^{AA'}\right)$, where $\D: \B(\H_{A'}) \to \B(\H_B)$ and $\F:\B(\H_{A'}) \to \B(\H_C)$ are two quantum processes (CPTP maps), and $\phi_{+}^{AA'}$ is the projection on     the maximally entangled state $|\phi_{+}^{AA'}\rangle=\frac{1}{\sqrt{d_A}} \sum_{i=1}^{d_A} |ii\rangle$, where $\{|i\rangle\}_{i=1}^{d_A}$ is an orthonormal basis for $A$   . Therefore, in this case the condition $ \id \otimes \E (\rho^{AB}) = \sigma^{AC}$ becomes equivalent to the \emph{degradability} of $\D$ into $\F$, that is, $\F = \E \circ \D$, and we denote it simply by $\F \prec_q \D$. Notice that notions equivalent to quantum majorization have previously been considered in Refs.~\cite{shmaya,chefles,FBCMP,BSD2014,Jenc2015} in the contexts of quantum statistics and quantum information theory.

\begin{figure}
	\centering
	\includegraphics[width=1\linewidth]{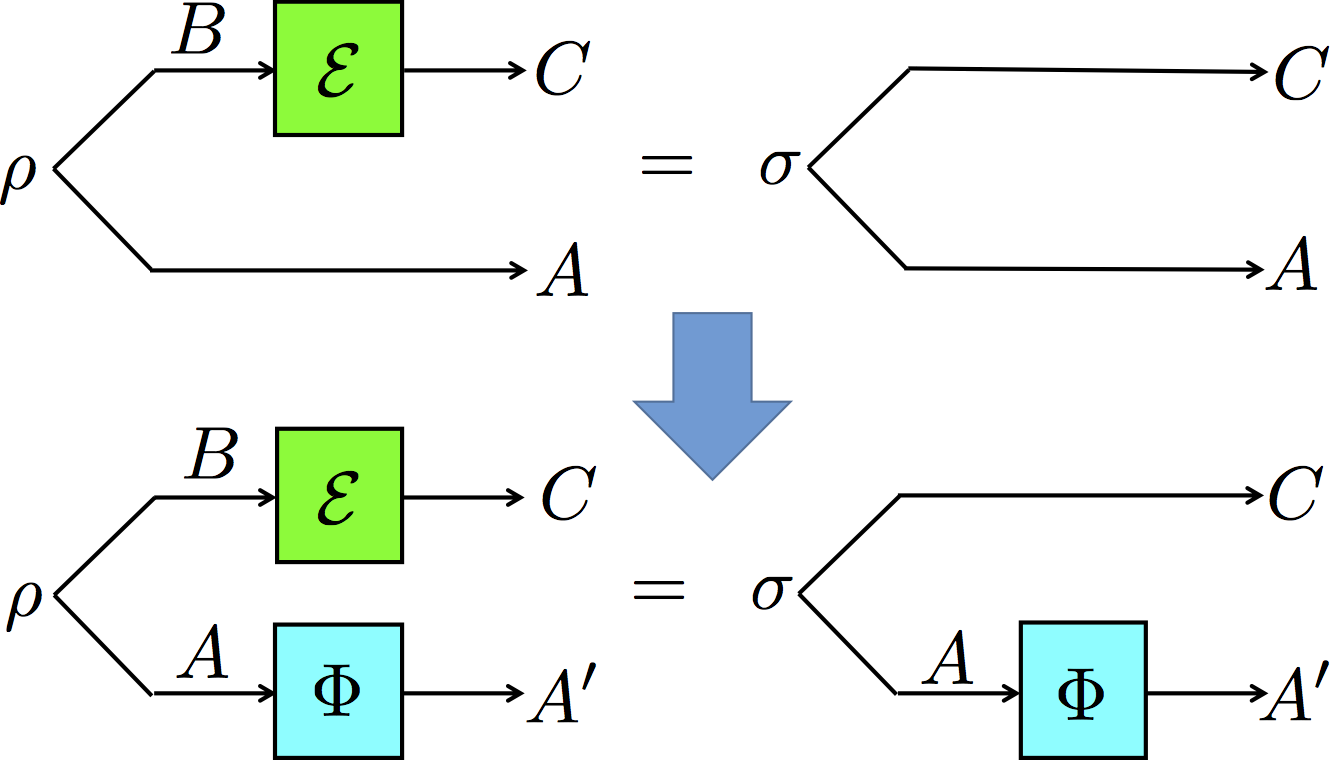}
	\caption{The condition of quantum majorization $\sigma^{AC}\prec_{q}\rho^{AB}$ implies the infinite set of relations $(\Phi\otimes\id)(\sigma^{AC})\prec_{q}(\Phi\otimes\id)(\rho^{AB})$, where $\Phi$ is any CPTP map acting on system $A$ (cfr. Eq.~(\ref{eq:necessary-conds}) in the main text). Theorem~\ref{mainthm} provides a complete set of monotones for quantum majorization, expressed as entropic functions of the bipartite state and the channel $\Phi$ acting on it.}
	\label{fig:unique-pic}
\end{figure}

Quantum majorization hence generalizes classical stochasticity and captures the notion that the process $\F$ is in some sense ``more disordered" than $\D$, since it can be obtained from $\D$ via $\E$. However, it does not say whether $\E$ is experimentally easy to perform. Typically, in resource theories it is important to place such additional restrictions, and demand that $\E$ is a ``free" operation of the theory. Many resource theories, such as entanglement theory, do not admit a simple specification, however, as we shall see shortly, in both the resource theories of asymmetry and thermodynamics, such a restriction of $\E$ to lie in a subset of free (symmetric or thermodynamic) processes can be made with a natural modification of our core result. 

\section{Characterization of Quantum Majorization}

Given the two bipartite states $\rho^{AB}$ and $\sigma^{AC}$, how can we determine if $\rho^{AB}$  quantum majorizes $\sigma^{AC}$? One simple and intuitive necessary condition, that follows from the data processing inequality, is that $$S(A|B)_{\rho}\leq S(A|C)_{\sigma}\;,$$ where $S(A|B)=S(A,B)-S(B)$ is the conditional entropy, and $S(\rho)=-\tr[\rho\log\rho]$ is the von-Neumann entropy. The intuition is that, if $\sigma^{AC} \prec_q \rho^{AB}$, then information about system $A$ is more accessible from system $B$ than from system $C$. Hence, the uncertainty of $A$ given $B$, i.e., $S(A|B)$, can only be smaller than the uncertainty of $A$ given $C$, i.e., $S(A|C)$. However, only one entropic condition is far from being sufficient to completely characterize quantum majorization.

In order to produce more necessary conditions, one can use a similar intuition to generate infinitely many necessary conditions that follows from the following observation (see Fig.~\ref{fig:unique-pic}): 
\be\label{eq:necessary-conds}
\sigma^{AC} \prec_q \rho^{AB}\implies\Phi\otimes\id \left(\sigma^{AC}\right) \prec_q \Phi\otimes\id \left(\rho^{AB}\right)
\ee
for any quantum process $\Phi: \B(\H_A) \to\B(\H_{A'})$. 
Note that $\Phi$ is acting on system $A$ while $\mE$ in Definition~\ref{maindef} is acting on system $B$.
We therefore conclude that if $\sigma^{AC} \prec_q \rho^{AB}$ then, for any quantum process $\Phi$, we must have:
\be\label{condent}
S(A'|B)_{\Phi\otimes\id\left(\rho^{AB}\right)}\leq S(A'|C)_{\Phi\otimes\id\left(\sigma^{AC}\right)}\;.
\ee
While the conditions above are necessary, again they are not sufficient, and even in the purely classical case: there exist \emph{classical} states $\rho^{AB}$ and $\sigma^{AC}$ such that $\sigma^{AC} \not\prec_q \rho^{AB}$, even though the above equation holds of all $\Phi$ (and any dimensions of $A'$)~\cite{KM1977,ElG1977,Buscemi}. 

On the other hand, in the following central result of our paper, we show that if one replaces the conditional (von-Neumann) entropy in~\eqref{condent} with the \textit{conditional min-entropy}~\cite{Ren05}, then the inequalities in~\eqref{condent} indeed provide, if all simultaneously satisfied, a sufficient condition for quantum majorization. Moreover, we can restrict $\Phi$ to be an entanglement breaking channel, and bound the dimension of system $A'$ to be no greater than the dimension of system $C$. Similar results, dubbed ``reverse data-processing theorems,'' have been obtained before~\cite{FBarxiv,Buscemi,Jenc2015,FB-nww}, although in a different framework involving extra ancillas and a classical reference system, while the present relations are fully quantum and do not need additional external systems.

The conditional min-entropy, $H_{\rm min} (A|B)_\Omega$, of a bipartite state $\Omega^{AB}$, is defined as~\cite{Ren05}
\be\label{eq:min-ent-def}
H_{\rm min} (A|B)_\Omega : = -\log\inf_{\tau^B \ge 0} \{ \Trr{\tau^B} : \I^A \otimes \tau^B \ge \Omega^{AB}\}.
\ee
It is known to be a single-shot analog of the conditional (von-Neumann) entropy. This analogy is particularly motivated by the fully quantum asymptotic equipartition property~\cite{Tom09}, which states that in the asymptotic limit of many copies of $\Omega^{AB}$, the smooth version of $H_{\min}(A|B)$ approaches the conditional (von-Neumann) entropy.  The conditional min-entropy has numerous applications in single-shot quantum information (e.g., Ref.~\cite{Toma-book} and references therein), quantum hypothesis testing (e.g.~\cite{Buscemi,BG16} and references therein), and quantum resource theories~\cite{Gour17}.

\begin{theorem}\label{mainthm}
Let $\rho^{AB}\in \B(\H_A \otimes \H_B) $ and $\sigma^{AC}\in \B(\H_A \otimes \H_C)$ be two compatible bipartite quantum states. Denote the dimension of any system $X$ as $d_X\in \mathbb{N}$. The following are equivalent:
\begin{enumerate}
\item The state $\rho^{AB}$ quantum majorizes $\sigma^{AC}$, 
\begin{equation}\label{bdef}
\sigma^{AC}\prec_q \rho^{AB} . 
\end{equation}
\item For any quantum process (CPTP linear map) $\Phi: \B(\H_A) \to\B(\H_{A'})$, with $d_{A'}=d_C$, 
\be\label{bhmin}
H_{\min}(A'|B)_{\Phi\otimes\id\left(\rho^{AB}\right)}
\leq H_{\min}(A'|C)_{\Phi\otimes\id\left(\sigma^{AC}\right)}
\ee
\item Eq.~\eqref{bhmin} holds for any measure-and-prepare quantum channel 
$\Phi:\B(\H_A)\to\B(\H_{A'})$ of the form: 
\be\label{formphi1}
\Phi\left(\eta^A \right)=\sum_{j=1}^{d_{A}^{2}}\tr\left[M_{j}^A\eta^A \right]\omega_{j}^{A'}\;,
\ee
where $\{M_{j}^A\}$ is an arbitrary, but fixed, informationally complete POVM on system $A$, while the states $\{\omega_{j}^{A'}\}$ can freely vary.
\item For any $\Phi:\B(\H_A) \to\B(\H_{A'})$ of the form~\eqref{formphi1} the following holds:
\be\label{eq2}
2^{-H_{\min}(A'|B)_{\Phi\otimes\id\left(\rho^{AB}\right)}}
\geq d_C\;\la\phi_+^{A'C}|\Phi\otimes\id\left(\sigma^{AC}\right)|\phi_+^{A'C}\ra
\ee
where $|\phi_+^{A'C}\>$ is a maximally entangled state between systems $A'$ and $C$.
\end{enumerate}
\end{theorem}

\begin{remark}
In the classical case, both $\rho^{AB}=\sum_{x,y}p_{xy}|x\lr x|\otimes|y\lr y|\equiv P$ and $\sigma^{AC}=\sum_{x,z}q_{xz}|x\lr x|\otimes|z\lr z|\equiv Q$ are diagonal, where $P$ (and $Q$) is the matrix whose components are the probabilities $p_{xy}$ ($q_{xz}$). Therefore, the relation $\sigma^{AC}=\id\otimes\mE\left(\rho^{AB}\right)$ can be expressed as $Q=SP$, where $S$ is a column stochastic matrix, so that $Q^T\prec_m P^T$~\footnote{The relation $Q=SP$ is equivalent to $Q^T=P^TS^T$, with $S^T$ being a row stochastic matrix.}. Dahl obtained in~\cite{Dahl1999} that $P$ matrix-majorizes $Q$ if and only if for all sub-linear functionals $f$, that can be written as a maximum of a finite number of linear functionals, the following holds:
\be
\sum_{j}f(\mbf{p}_j)\geq \sum_{k}f(\mbf{q}_k)\;,
\ee
where  $\mbf{p}_j$ and $\mbf{q}_k$ are the rows of $P$ and $Q$, respectively. Since classically $2^{-H_{\min}(A|B)}$ is a sub-linear functional (see more details in the supplementary material), our theorem above provides the same result for the classical case, with a slight improvement that $f$ can be restricted to sub-linear functionals that can be written as a maximum of at most $d_C$ linear functionals.
\end{remark}

\begin{remark}
The conditions in Eqs.~(\ref{bhmin},\ref{formphi1}) are given in a form of monotones; i.e. functions that behave monotonically under certain operations (in our case under quantum majorization). In quantum resource theories monotones quantify resources as they do not increase under free operations. As we will see below, the conditional min-entropies that appear in Theorem~\ref{mainthm} can be used to quantify asymmetry in the resource theory of quantum reference frames, and athermality in quantum thermodynamics. Since Eq.(\ref{bhmin}) (or Eq.(\ref{eq2})) has to hold for any set of density matrices $\{\omega_{j}^{A'}\}$, quantum majorization is characterized in Theorem~\ref{mainthm} with infinite number of monotones. This is a consequence of the fact that we consider exact transformations, and typically an exact (algebraic) solution to such an SDP feasibility problem is NP-hard. However, we show in the supplemental material that if we allow for a small error, then determining if $\rho^{AB}$ quantum majorizes $\sigma^{AC}$ can be solved efficiently using semidefinite programming.
\end{remark}

If only system $A$ is classical, that is
the states $\rho^{AB}=\sum_{i}p_i|i\lr i|\otimes\rho_i$ and $\sigma^{AC}=\sum_{i}p_i|i\lr i|\otimes\sigma_i$ are classical-quantum states,
we get that~\eqref{bdef} is equivalent to 
\be\label{sc}
\sigma_i=\mE(\rho_i)
\ee 
for all $i$ such that $p_i>0$. This is a classic problem in quantum hypothesis testing, and the results presented here complement previous results in the same direction~\cite{AU1980,chefles,FBCMP,Matsu2010,Buscemi,Jenc2015,Teiko2012}. In particular, it can be shown (see Supplementary Material) that the theorem above implies the following corollary:
\begin{corollary}
	There exists $\mE$ satisfying~\eqref{sc} if and only if for any set of $n$ density matrices $\{\omega_{i}^{A}\}_{i=1}^{n}$ we have
$H_{\min}(A|B)_{\Omega}\leq H_{\min}(A|C)_{\Omega}$, where
\be
\Omega^{ABC}=\frac{1}{n}\sum_{i=1}^{n}\omega_{i}^{A}\otimes\rho_{i}^{B}\otimes\sigma_{i}^{C}\;.
\ee
The same relation holds if the uniform distribution $1/n$ is replaced with any other arbitrary distribution $q_i$, with the only condition that $q_i>0$.
\end{corollary}

\section{A complete set of entropic conditions for the Resource Theory of Asymmetry}
  
So far we considered the relation $\sigma^{AC}=\id\otimes\mE\left(\rho^{AB}\right)$ with arbitrary CPTP map $\mE:\B(\mH_B)\to\B(\mH_C)$.
We now impose additional constraint on $\mE$, requiring it to be $G$-covariant with respect to a compact group $G$. That is, $\mE$ is $G$-covariant with respect to two unitary representations of $G$ on systems $B$ and $C$, denoted respectively by $\{V_{g}\}_{g\in G}$ and $\{U_{g}\}_{g\in G}$, if
\be
U_{g} \E (\rho) U_{g^{-1}} = \E (V_{g} \rho V_{g^{-1}})\quad\forall\;g\in G\;.
\ee
We write $\sigma^{AC}\prec_{q}^{G}\rho^{AB}$, if $\sigma^{AC}=\id\otimes\mE\left(\rho^{AB}\right)$ with a $G$-covariant CPTP map $\mE$.

Theorem~\ref{mainthm} can be easily upgraded to accommodate $G$-covariant maps. Particularly,
it can be shown that $\sigma^{AC}\prec_{q}^{G}\rho^{AB}$ if and only if
\be
H_{\min}(A'|B)_{\mG\left[\Phi\otimes\id\left(\rho^{AB}\right)\right]}\leq H_{\min}(A'|C)_{\mG\left[\Phi\otimes\id\left(\sigma^{AC}\right)\right]}
\ee
for all CPTP entanglement breaking maps $\Phi$ of the form~\eqref{formphi1}. Here $\mG:\B\left(\mH_{A'}\otimes\mH_{C}\right)\to\B\left(\mH_{A'}\otimes\mH_{C}\right)$ is the bipartite $G$-twirling map given by
\be
\mG[\tau^{A'C}]=\int dg\;(\overline{U}_g\otimes U_g)\;\tau^{A'C}\;(\overline{U}_{g}^{\dag}\otimes U_{g}^{\dag})\;,
\ee
where the over bar denotes the complex conjugation made with respect to the basis fixed by the choice of the maximally entangled state $|\phi_{+}^{A'C}\>$.

In the special case in which both $\rho^{AB}=|0\lr 0|^A\otimes\rho^B$ and $\sigma^{AC}=|0\lr 0|^A\otimes\sigma^C$ are product states, our theorem is simplified to the following statement: $\rho^B$ can be converted to $\sigma^C$ by a $G$-covariant map if and only if for any density matrix $\eta^{A'}$, 
\be
H_{\min}(A'|B)_{\mG[\eta^{A'}\otimes \rho^{B}]}\leq H_{\min}(A'|C)_{\mG[\eta^{A'}\otimes\sigma^{C}]}\;.
\ee
 Therefore, the quantities $H_{\min}(A'|B)_{\mG[\eta^{A'}\otimes \rho^{B}]}$, for varying reference state $\eta^{A'}$, provide a complete set of asymmetry monotones for the resource theory of asymmetry \cite{gour2008resource,spekkens2009relative, marvian2013modes, ahmadi2012way, marvian2013asymmetry,marvian2014extending, marvianthesis, cirstoiu2017}.   In other words, for any given state $\eta^{A'}$, the function $H_{\min}(A'|B)_{\mG[\eta^{A'}\otimes \rho^{B}]}$ provides a single-copy quantification of the amount of asymmetry of state $\rho^{B}$ relative to the symmetry group $G$.
We are now ready to discuss the application of this result to quantum thermodynamics.

\section{A complete set of entropic conditions for Quantum Thermodynamics}

While thermodynamics in macroscopic, equilibrium, and classical regimes is well understood \cite{fermi1956thermodynamics, callen}, there is the fundamental question of how one {  can extend thermodynamic notions into non-equilibrium, finite-sized systems \cite{jarzynski, crooks,campisi, reeb}, and in particular systems displaying highly non-classical properties such as quantum coherence, contextuality, and entanglement \cite{partovi,rio2010thermodynamic,jennings2010entanglement, modi, goold,vinjanampathy}.} One particular approach to this problem \cite{janzing2000thermodynamic, horodecki2013fundamental, aberg2013truly, brandao2011resource,brandao2013second, gour-mueller, narasimhachar,lostaglio2015description, cwiklinski2014limitations,lostaglio2014quantum,lostaglioStochastic, weilenmann, korzekwa2015extraction,chiribella2015, mueller2017} has been to utilize tools and concepts developed in the study of entanglement, which is understood within the framework of resource theories. A resource theory provides a way to quantify physical characteristics that are not simply given by Hermitian observables, and is defined once we specify a set of \emph{free states}, as those that do not have the properties one wishes to study, together with set of \emph{free operations}, that are compatible with the set of free states in the sense that their action on any free state always yields another free state. 

This approach of analysing thermodynamics in terms of its process structure (instead of starting with problematic terms such as `heat' or `work' or `entropy') turns out to have a long and successful history dating back to the 1909 seminal work of Carath\'eodory \cite{caratheodory1909}. Other notable accounts were obtained in 1964 by Giles \cite{giles1965mathematical} and more recently in 1999 by Lieb and Yngvason \cite{lieb1999physics}, who provided a thorough analysis in terms of adiabatic accessibility. Moreover, it has recently been shown in \cite{weilenmann} that the thermodynamic structure of incoherent quantum states obtained from an information-theoretic perspective coincides with the phenomenological analysis in \cite{lieb1999physics}, which demonstrates the soundness of the resource theoretic approach.

In thermodynamics a preferred class of states are singled out as free states from the condition of complete passivity \cite{pusz1978, lenard1978}. In the simplest case the Gibbs state $ \frac{1}{Z} e^{-\beta H}$, with $\beta =(kT)^{-1}$ and $Z = \tr [ e^{-\beta H}]$, is the only quantum state that can be freely admitted without trivialising the theory energetically. More generally, in the presence of additional   additive conserved charges $\{X_1, \dots X_n\}$, such as angular momenta and particle numbers,     this can be extended (under certain assumptions on external constraints \cite{callen,bailan,halpern2014, vaccaro2011, halpern2016,lostaglio2015,halpern, guryanova}) to the generalized Gibbs state 
\begin{equation}\label{gge}
\gamma^A = \frac{1}{\mathcal{Z}} e^{-\beta (H^A - \sum_k \mu_k X^A_k)},
\end{equation} 
with $\{\mu_k\}$ being Lagrange multiplier constants for the conserved quantities and $\mathcal{Z}=\tr [e^{-\beta (H^A - \sum_k \mu_k X^A_k)}]$. In the case that we just have a single additional number operator $N$, the constant is the usual chemical potential \cite{callen}.

\subsection{Generalized Thermal Processes}

Our thermodynamic framework is an extension of the resource theory of Thermal Operations (TOs) \cite{janzing2000thermodynamic, horodecki2013fundamental, brandao2011resource} to a set of transformations that contains TOs as a proper subset. It is an extension in two ways: firstly it makes a weaker assumption about the underlying microscopic process, and secondly it is defined in terms of a collection of distinguished thermodynamic observables, such as those in the Generalized Gibbs ensemble, and not just in terms of energy. We shall refer to these free transformations as \textit{(generalized) thermal processes} (abbreviated to TPs), and they are specified by the following three physical assumptions:

\begin{enumerate}[{A}1]
\item \textbf{(Microscopic conservation)} Each input quantum system and output quantum system has a Hamiltonian $H$,  and a collection of distinguished observables $X_1, \dots X_n$. The total energy and the observables $\{X_k\}$ are conserved microscopically in any free process, and moreover $[H,X_k] = 0$ for all $k=1, \dots , n$.
\item \textbf{(Equilibrium preservation)} For every (input or output) system $A$, an equilibrium free state exists that is stable under the class of free processes.
\item \textbf{(Incoherence)} The free processes do not exploit any  sources of quantum coherence between eigenbases of conserved quantities.
\end{enumerate}

Assumption (A1) ensures that every quantum system $A$ has a well-defined Hamiltonian $H^A$ at the initial time and some other Hamiltonian $H^{A'}$ at the final time. It also allows for an arbitrary set of additional conserved charges, as discussed. More precisely, any TP map $\E$ on $A$ admits a Stinespring dilation onto some larger system $B$ such that
\begin{equation}\label{Stinespring}
\E(\rho^A) = \tr_C V (\rho^A \otimes \sigma^B ) V^\dagger
\end{equation}
where $B$ is some other quantum system defining the thermal environment. The assumption (A1) implies that the isometry $V$ obeys
\begin{align}\label{conservation-law}
V ( H^A \otimes \I^B + \I^A \otimes H^B) &= ( H^{A'} \otimes \I^C + \I^{A'} \otimes H^C)V \nonumber\\
V ( X^A_k \otimes \I^B + \I^A \otimes X^B_k) &= ( X^{A'}_k \otimes \I^C + \I^{A'} \otimes X^C_k)V
\end{align}
for all $k=1, \dots n$, which defines the microscopic energy conservation and the conservation of the charges. Note that we also allow the input system and output to differ, which may occur due to the presence of strong-couplings that affect factorizability into independent subsystems. It is also important to emphasize that we do not assume or require microscopic control of $V$. It is only the total process $\E$ that is experimentally relevant. The particular set of observables are determined by the physical context and we shall refer to them as the thermodynamic observables for the system.

Assumption (A2) says that for every system $A$ there is a state $\rho^A_\star$, such that $\E(\rho^A_\star) = \rho^A_\star$ for all TPs $\E$. However (A1) singles out a set of distinguished observables $\{H^A, X^A_1, \dots ,X^A_n\}$ that microscopically are additively conserved. The fact that $\rho_\star^A$ is a free state of the theory implies \cite{lostaglio2015,halpern, guryanova} that the only form of $\rho^A_\star$ that can yield a non-trivial resource theory in these observables is one for which $\log \rho^A_\star$ is a linear combination of the observables -- namely it must be a generalized Gibbs state $\gamma^A$ as defined in (\ref{gge}), at some fixed temperature $T= (k\beta)^{-1}$ and Lagrange multipliers $\mu_1, \dots, \mu_n$. Therefore the free states of the theory are defined uniquely by these parameters.

The final assumption (A3) is a statement of non-classicality within the theory and requires us to provide an explicit accounting for coherence resources. It is known for thermal operations that if the only coherences present are within energy eigenspaces then the resultant theory is essentially classical, and is described by thermo-majorization \cite{horodecki2013fundamental}. However coherences between energy eigenspaces behave differently and do not have such a classical description \cite{lostaglio2015description}. Therefore one must carefully account for these coherences thermodynamically. The precise formulation of this requirement in the case of energy is that if any free process $\E$ is obtained as in~\eqref{Stinespring}
for some microscopically conserving interaction $V$ with $B$, then $\E$ must also be realisable as
\begin{equation}
\E(\rho) = \tr_C W (\rho^A \otimes \widetilde{\sigma}^B ) W^\dagger
\end{equation}
where $W$ is again a conserving interaction, and $\widetilde{\sigma}^B =\lim_{\tau \rightarrow \infty} \frac{1}{\tau}\int_0^\tau dt U_B(t) \sigma^B U^\dagger_B(t)$ is the state $\sigma^B$ after a complete dephasing of coherences between energy eigenbases. This captures the notion that $\E$ is realized without consuming any coherent resources from the external degrees of freedom in $B$. At the level of quantum operations on $S$, this implies that we have the following symmetry property for all free operations 
\begin{equation}\label{cov}
U'(t) \E (\rho^A) U'^\dagger(t) = \E (U(t) \rho^A U^\dagger(t))
\end{equation}
where $U(t) = \exp [-i tH^A]$ and $U'(t) = \exp[ -i t H^{A'}]$ are respectively \emph{free} evolution of the input/output system for an interval of time $t$. The operation $\E$ is said to be covariant under time-translation. The more general case of multiple conserved charges is discussed below.

The three physical assumptions specify the set of \emph{generalized thermal procesess}, and it is readily seen that it contains the set of thermal operations. In the case when the only conserved quantity is $H$, there is no particular physical reason to choose one set of operations over the other. However, in the case of multiple conserved charges $X_1, \dots , X_n$, the use of TPs has an advantage in that it allows one to handle generalized Gibbs ensemble scenarios more easily. The details of system $B$ are, in general, not observed thermodynamical degrees of freedom, and with an explicit microscopic specification, such as with thermal operations, subtleties arise in the case of additional charges. Particularly, subtleties arise if one wishes to have non-trivial $\mu_k$ Lagrange multipliers in the generalized Gibbs ensemble (\ref{gge}) and also satisfy assumption (A1). The formulation here simply avoids this by not demanding a specific form for the microscopic state $\tilde{\sigma}^B$ in the definition of the free processes. Assumption (A3) only constrains the microscopic details to the extent that there are no observable effects of coherence at the level of the process $\E$.

In the Supplementary Material we show that our core result on quantum majorization can be adapted to the setting of generalized thermal processes to fully describe the state interconversion structure. This is obtained by establishing the following lemma, which is proved in the Supplementary Material.

\begin{lemma} 
Consider two sets of thermodynamic observables $\{H^S, X^S_1, \dots, X^S_n\}$ for quantum system $S=A$ and quantum systems $S=A'$. Then, the set of all quantum processes from $A$ into $A'$ defined by (A1-A3) coincides with the set of all $\gamma$-preserving processes on $A$ that are covariant under the group $G$ generated by the thermodynamic observables on $A$ and $A'$.
\end{lemma}

\subsection{State conversions under thermal processes} 

Since TPs are $G$-covariant we may make use of our earlier results on $G$-covariant state interconversion of a collection of states $\{\rho_i^A\}$ into $\{\sigma_i^B\}$. We first consider the case where energy is the only distinguished thermodynamic observable that is conserved microscopically. Combining the $G$-covariant version of Theorem~\ref{mainthm} with the above lemma we get the following theorem (see Supplementary Material for more details).

\begin{theorem}\label{thermo-1}Let $A$ and $A'$ be two quantum systems, with the respective Hamiltonians $H^A$ and $H^{A'}$ being the only thermodynamic observables, and let $0<q<1$ be a fixed number. The state transformation $\rho^A \longrightarrow \sigma^{A'}$ is possible under generalized thermal processes at a temperature $T=(k\beta)^{-1}$ if and only if for all reference frame systems $R$ with the same dimension as of $A'$ and with Hamiltonian $H^R=-\left(H^{A'}\right)^T$, and for all pairs of states $\boldsymbol{\eta} = (\eta^R_1, \eta^R_2)$, we have 
\be
S_{\boldsymbol{\eta}} (\rho^A) \le S_{\boldsymbol{\eta}}(\sigma^{A'})\;,
\ee
 where $S_{\boldsymbol{\eta}}(\rho^A) := H_{\rm min} (R|A)_\Omega$ and 
\begin{equation}
\Omega^{RA} =\Big\langle q\eta^R_1 \otimes \rho^A+ (1-q)\eta^R_2 \otimes \gamma^A) \Big\rangle\ .
\end{equation}
Here, $\gamma^A = \exp[-\beta H^A]/Z$ is the Gibbs state on $A$,  and $\langle\omega^{RA}\rangle\equiv \lim_{\tau\rightarrow \infty} 1/\tau \int_0^\tau\!dt\ U(t)\omega^{RA} U^\dag(t)$  is the channel that maps any state of $RA$ to its time-averaged version, and $U(t)= \exp [ -i t (H^R \otimes \I^A + \I^R \otimes H^A)]$ is the unitary time-evolution under the Hamiltonian for the composite system $RA$.
\end{theorem}
  
It is important to note that these conditions can be greatly reduced. In particular one can simply consider $q=\frac{1}{2}$ alone, however in some cases it is useful to choose different values and so we give the general case here. Also, it readily seen that the state $\eta^R_2$ can be chosen to be block-diagonal in the energy eigenbasis, while $\eta_1^R$ can be restricted to reference frame states that have the same modes of coherence as $\rho^A$ \cite{marvian2013modes,lostaglio2014quantum}. 

These conditions have a range of physical implications and describe the features of quantum thermodynamics in a compact way.   A key obstacle in quantum thermodynamics is that to determine the existence of the transformation  $\rho^A \longrightarrow \sigma^{A'}$, one needs to consider two different types of physical properties of states: (i) properties related to their energy distribution, which leads to conditions such as thermo-majorization  \cite{ruch-schranner}, and (ii) properties related to  the coherence in the energy eigen-basis.  Roughly speaking, one needs to check that the initial state $\rho^A $ has (at least) as much as free energy and coherence as the desired final state $\sigma^{A'}$.  

It is not possible in general to quantify both of these simultaneously in a measurement scheme.  Coherences in energy are precisely the time-dependent components of a quantum system and thus one encounters an obstacle of complementarity between time and energy measurements. Physically these two aspects can be viewed as `clock' and `work' regimes of a quantum system.   Theorem \ref{thermo-1}  gets around this complementarity  by allowing the reference system $R$ to act simultaneously as a `clock/work reference'. In other words, one can interpolate smoothly between the two regimes via the different choices of quantum states $\eta^R$. This is illustrated schematically in Figure \ref{Fig3}.  

\begin{figure}
	\centering
	\includegraphics[width=1\linewidth]{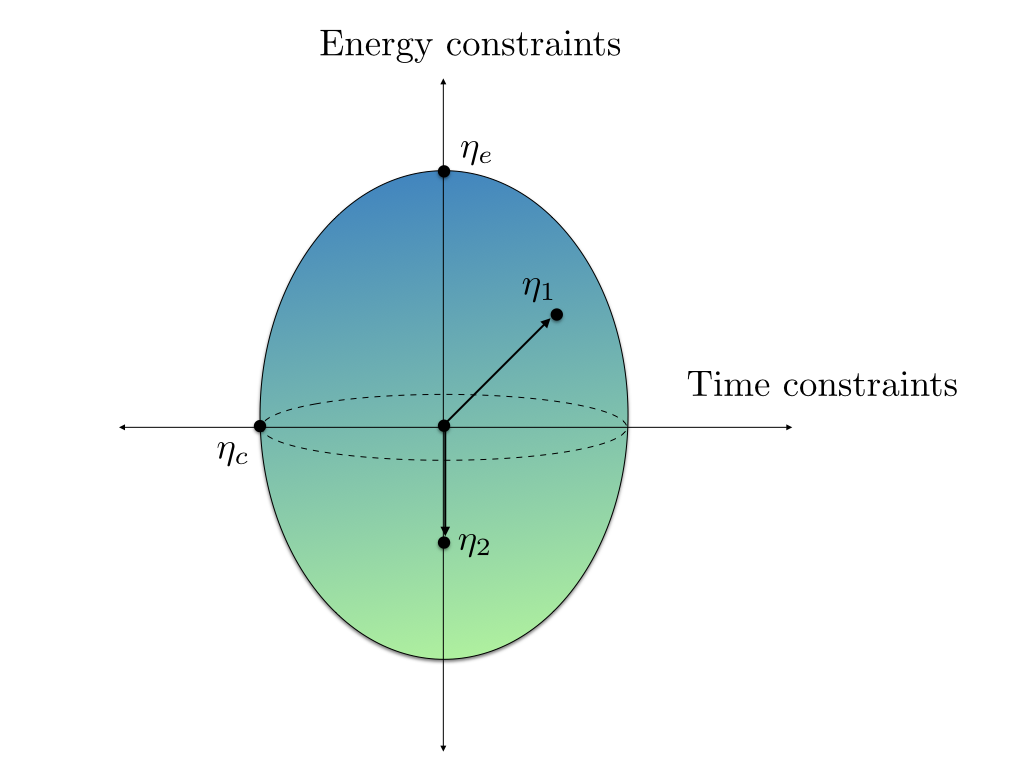}
	\caption{\textbf{Time--energy constraints for Thermal Processes.} The entropic conditions for a state transformation $\rho^A \longrightarrow \sigma^{A'}$ under TPs are defined with respect to a quantum reference frame $R$ and two states $\eta_1^R$ and $\eta_2^R$. The schematic vertical axis denotes states block-diagonal in energy (e.g. an energy eigenstate $\eta_e= |E\rangle\langle E|$), while the horizontal axis denotes states with maximal time-dependent oscillations -- `clock' states $\eta_c$ of $R$. When $\eta_1^R$ is confined to being incoherent (the vertical axis) we recover thermo-majorization. For $R$ being macroscopic and $\eta_1^R = \eta_c$ we obtain a Page-Wootters clock constraint on the thermodynamic transformation. Varying $\eta_1^R$ smoothly interpolates between the time constraints and energy constraints.
}
	\label{Fig3}
\end{figure}

To see this better, we first consider the case where either the input or output state is incoherent in the energy eigenbasis. This regime is described by an essentially classical stochastic energy condition. The following result is shown in the Supplementary Material.
\begin{corollary}\label{thermo-major} Let $A$ and $A'$ be two quantum systems, with respective Hamiltonians $H^A$ and $H^{A'}$ being the only thermodynamic observables. Let $\rho^A$ and $\sigma^{A'}$ be quantum states on the input and output systems, respectively. If either $[\rho^A, H^A]=0$ or $[\sigma^{A'}, H^{A'}]=0$, then the state transformation $\rho^A \longrightarrow \sigma^{A'}$ is possible under generalized thermal processes at a temperature $T=(k\beta)^{-1}$ if and only if $[\sigma^{A'}, H^{A'}] = 0$ and $\<\rho^A\>$ thermo-majorizes $\sigma^{A'}$.
\end{corollary}
This recovers previous results \cite{horodecki2013fundamental} on quantum thermodynamics for the case of one of the states having no coherences between energy eigenspaces.
Moreover, in the case of incoherent input $\rho^A$, 
the use of a coherent reference state $\eta^R$  does not yield  any additional constraint. Specifically, $\Omega^{RA} = q \<\eta^R_1\> \otimes \rho^A +(1-q) \<\eta^R_2\> \otimes \gamma^A$, and so the coherence of states $\eta_1^R$ and $\eta_2^R$ is irrelevant. The only relevant constraints in state transformation $\rho^A\rightarrow\sigma^{A'}$ are constraints related to the energy distribution of states.

On the other hand, if both the input-output  states $\rho^A$ and $\sigma^{A'}$ contain coherence, then  by choosing reference states $\eta_1^R$ which contain coherence, we obtain  new additional coherence constraints, i.e. constraints independent of thermo-majorization. Note that coherence with respect to energy eigenbasis is equivalent to symmetry-breaking (asymmetry) with respect to time-translations generated by the system Hamiltonian.  In other words, coherence of  states $\rho^A$ and $\sigma^{A'}$ is related to how well time $t$ can be estimated from states $\rho^A(t)=e^{-i H_A t}\rho^A e^{i H_A t}$ and $\sigma^{A'}(t)=e^{-i H_{A'} t}\sigma^{A'} e^{i H_{A'} t}$. 

The TPs are both covariant under time-translation and preserve the Gibbs state. In the Supplementary Material we will show that the converse is also true (i.e. a covariant Gibbs preserving map is a TP). Therefore, previously discussed measures, such as those that are based on Renyi Divergences of the form $A_\alpha(\rho^A) = S_\alpha(\rho^A || \< \rho^A\>)$,  behaves monotonically under TPs, and provide independent thermodynamic constraints beyond thermo-majorization \cite{lostaglio2015description}. One can also use constraints on modes of coherence \cite{lostaglio2014quantum} and the Fisher Information \cite{hyukjoon-2017}, to derive other independent measures of athermality.  However, the set of conditional min-entropy measures obtained here is complete and therefore sufficient to imply the monotonicity of all of these measures. 

In the Supplementary Material we show that the entropic conditions with $\eta_1^R$ being incoherent in energy leads to thermo-majorization, and captures the degree to which the system $A$ is ordered in energy. Since in quantum systems one has complementarity between time and energy one might expect that the case of $\eta_1^R$ being highly coherent in energy might therefore capture the degree to which $A$ is ordered in some temporal sense.

This turns out to be the case, although since time forms a continuous one-parameter group there are technical obstacles to making this statement precise. However, as we show in the Supplementary Material, one can in general make finite precision approximations and model time evolution for any finite dimensional quantum system (which can be assumed to have an energy spectrum of rational numbers and thus has periodic dynamics under its Hamiltonian) with the discrete group $\mathbb{Z}_N$, for some sufficiently large $N$ and with $t = n \epsilon$. Here $\epsilon>0$ is the minimal time interval that can be resolved. The representation of this discrete group on $A$ is given by $n \mapsto U^A_\epsilon(n) := \exp [ - i n\epsilon H^A]$ and so the system is modelled as evolving in discrete time steps. Under these approximations one can replace (A3) with a slightly weaker version (A3$'$) described in the Supplementary Material, and the interconversion conditions can be repeated for $G=\mathbb{Z}_N$ instead.

We define \emph{clock-times} as the discrete instances $t=0, \epsilon, \dots, n\epsilon, \dots, (N-1)\epsilon$ for the joint system $\H_R \otimes \H_A$. As shown in the Supplementary Material, there exist reference frame systems $R$ that can provide a perfect classical encoding of the clock times into quantum states $\{|0\>^R, |1\>^R, \dots, |N-1\>^R \}$, and for which $U^R_\epsilon(n) |0\>^R = |n\>^R$  for any $n$. Moreover these \emph{clock states} are built from uniform superpositions in the energy eigenstates of $R$, and so are in a sense ``maximally" coherent in energy. Given this, we can now demonstrate the claimed complementarity between time and energy and how it relates to the state of the reference $R$. We choose $\eta_1^R = |0\>\<0|^R$ and consider the limit $q\rightarrow 1$, which corresponds to the condition of time-translation covariance alone. For this one can show that 
\begin{equation}
\Omega^{RA} = \frac{1}{N} \sum_{k=0}^{N-1} |k\>\<k|^R \otimes \rho^A(n),
\end{equation}
where $\rho^A(n) := U^A_\epsilon(n) \rho^A (U^A_\epsilon(n))^\dagger$ is the state of the system $A$ at the $n^{\rm th}$ clock time for the joint system.
Now, since $\Omega^{RA}$ is a classical-quantum state, we have that~\cite{Kon09}
\begin{equation}
H_{\rm min} (R|A)_\Omega = -\log p_{\rm guess},
\end{equation}
where $p_{\rm guess}$ is the optimal Helstrom guessing probability for the ensemble of states $\{(\frac{1}{N}, \rho^A(n))\}_{n=0}^{N-1}$ on $A$. This implies that $2^{-H_{\rm min}(R|A)_\Omega}$ is the optimal guessing probability of the clock time $t=n\epsilon$ for the joint system, given the single copy of $\rho^A$. Monotonicity of $H_{\rm min}(R|A)_\Omega$ under the thermal processes therefore implies monotonicity of the clock time guessing probability for the system. Phrased differently, the time-translation covariance property of thermal processes implies that the ability of the thermodynamic system $A$ to act a quantum clock \cite{peres-clock, connes, gambini, malabarabra, moreva, erker, rankovic, clocks-cloners} can never increase. This demonstrates how the reference frame system $R$ functions to define both time and energy constraints on the state interconversion for the system $A$. 

We note that this result connects with foundational work by Page and Wootters \cite{page}, who considered how one can have dynamics in a universe that is covariant in time. They proposed a conditional probability formalism, which mirrors our present set up and relies on covariant measurements with $P(X^R=x | Y^A =y)$, the probability that some observable $X^R$ has a sharp value given a measurement of $Y^A$ yielding a particular result. These relational expressions were shown to describe dynamics within the time-translation invariant global state, such as $\Omega^{RA}$ here. 
 
Finally, we can state the necessary and sufficient conditions for the case of having additional, additively conserved observables $\{X_1, \dots ,X_n\}$. In this case assumption (A3) follows a similar argument to the one for energy, and the auxiliary system can be assumed to be in a state $\sigma^B$ for which   $\sigma^B = e^{-i s X^B_k} \sigma^B e^{is X^B_k}$, for all $s\in \mathbb{R}$ and for any thermodynamic observable $X^B_k$. Ranging over all the observables, this condition can be expressed more compactly as $\sigma^B =U(g) \sigma^B U^\dagger(g)$, for all unitary transformations $U(g)$ in the Lie group $G$ generated by the observables $\{H^B,X^B_1, \dots, X^B_n\}$. Note that this condition is equivalent to $\sigma^B =\int_G dg\ U(g) \sigma^B U^\dagger(g)$, where $dg$ is the uniform (Haar) measure over this group. Therefore, this assumption, together with (A1) imply that the process $\E$ is \emph{covariant} with respect to group $G$, i.e.  $\U_g \circ \E = \E \circ \U_g$ where $\U_g(\rho^A):= U(g) \rho^A U^\dagger(g)$.    In other words, the process is covariant under the symmetry group action generated by the thermodynamic observables on the input/output systems. Our main result on the thermodynamic structure of states under TPs is as follows.

\begin{theorem}[Generalized Thermal Processes] Let $A$ and $A'$ be two quantum systems, with thermodynamic observables $\{H^A, X^A_1 \dots ,X^A_n\}$ and $\{H^{A'}, X^{A'}_1, \dots X^{A'}_n\}$, respectively, and fix $0<q<1$. The state transformation $\rho^A \longrightarrow \sigma^{A'}$ is possible under generalized thermal processes at a temperature $T= (k\beta)^{-1}$ and at fixed Lagrange multipliers $\mu_1, \dots \mu_n$, if and only if for all reference frame systems $R$ of equal dimension to $A'$ with thermodynamic observables $H^R= - (H^{A'})^T$ and $\{X^R_k =-(X^{A'}_k)^T\}_{k=1}^n$ , and for all pairs of states $\boldsymbol{\eta}=(\eta_1,\eta_2)$ we have $S_{\boldsymbol{\eta}} (\rho^A) \le S_{\boldsymbol{\eta}}(\sigma^{A'})$, where $S_{\boldsymbol{\eta}}(\rho^A) := H_{\rm min} (R|A)_\Omega$ and
\begin{equation}
\Omega^{RA} = \int_G \!\! dg \,\, U(g) ( q\eta^R_1 \otimes \rho^A+(1-q) \eta^R_2 \otimes \gamma^A) U(g)^\dagger
\end{equation}
where $\{U(g)\}$ is the symmetry group generated by the additively conserved observables $\{H^R\otimes \I^A + \I^R \otimes H^A, X_k^R \otimes \I^A + \I^R\otimes X_k^A; \,\, k=1,\dots , n\}$ on the composite system $RA$, with group parameters $g$, and $\gamma^A = \exp [- \beta(H^A - \sum_k \mu_k X^A_k)]/\mathcal{Z}$, being the generalized Gibbs ensemble on $A$.
\end{theorem}
This result is a fully covariant statement that is based on minimal assumptions (A1-A3), and which reduces to Theorem (2) in the case of no additional thermodynamic observables beyond the system's energy.

\section{Conclusion}
In this work we introduced a new generalization of majorization for quantum processes, found a necessary and sufficient condition for this notion of majorization in terms of entropic quantities, and  demonstrated some of its applications in the context of the resource theories of asymmetry and quantum thermodynamics. In particular, we derived a complete set of entropic conditions for state transformations in both of these resource theories. In contrast to the previous results, which are  only applicable to restricted families of states (such as incoherent states) our approach can be applied to \emph{all} states. Furthermore, these results can be readily generalized to the case of approximate transformations in which we only require transformations up to an epsilon smoothing. 

Since these entropic monotones provide a full characterization of the resource, it is interesting to study their  operational interpretations.  We discussed some of these interpretations in the context of clocks. Another possible interpretation could be provided by the results of \cite{faist-minimal}, which relates  the smoothed   entropy $H^\epsilon_{\rm max}(R|A')$ to the minimal work cost to perform a quantum process. The duality relation between min and max entropies tells us that $H^\epsilon_{\rm max} (R|A') = - H^\epsilon_{\rm min} (R|C)$, where $C$ purifies the state on $RA'$, and so this suggests a potential interpretation of our results in terms of generalized work costs on a purifying environment.

We also introduced a new framework for quantum thermodynamics based on the notion of generalized thermal processes, which extends thermal operations, and is based on natural physical principles. 
This  explicitly handles coherences and is the first framework of its kind for which a complete set of state conditions has been derived.

\begin{acknowledgements}

GG would like to thank Rob Spekkens for pointing out Ref.~\cite{Dahl1999}, and for many discussions on classical matrix majorization. We would also like to thank Eric Chitambar, Philippe Faist, Mark Girard, Kamil Korzekwa and Matteo Lostaglio for useful discussions, and to Nicole Yunger Halpern for many comments on the first version. F.B. acknowledges support from JSPS KAKENHI, grants no. 26247016 and
no. 17K17796. GG acknowledges support from the Natural Sciences and Engineering Research Council of Canada (NSERC), DJ is funded by the Royal Society.

\end{acknowledgements}

\begin{titlepage}
\center{\large\textbf{Supplementary Material\\Quantum majorization and a complete set of entropic conditions\\for quantum thermodynamics}\\}
\center{~{ }\\}

\end{titlepage}

\appendix


This appendix is organized as follows. In appendix A we prove Theorem 1. In appendix B we discuss the classical case, and re-derive thermo-majorization from our conditions. In appendix C we prove the G-covariant version of Theorem 1, and also discuss the Stinespring dilations of covariant maps. In appendix D we discuss the generalized thermal processes, and show that they are equivalent to Gibbs Preserving Covariant (GPC) maps. Finally, appendix E we provide the proof of the norm expression in Eq.(23).

\section{Proof of Theorem~\ref{mainthm}}
In order to prove Theorem~\ref{mainthm}, we will begin by proving the following lemma. We recall the definition of conditional min-entropy in Eq.~(\ref{eq:min-ent-def}).

\begin{lemma}\label{mainlemma}
Let $\{\rho_i^B\}_{i=1}^{n}$ and $\{\sigma_{i}^C\}_{i=1}^{n}$ be two sets of $n$ density matrices in $\B(\mH_B)$ and $\B(\mH_C)$, respectively. Let $\{q_i\}_{i=1}^{n}$ be some arbitrary but fixed probability distribution with $q_i>0$.
For any set of $n$ density matrices $\{\omega_{i}^A\}_{i=1}^{n}$ in $\B(\mH_{A})$ (with $d_A=d_C$) define the following tripartite separable matrix:
\be
\Omega^{ABC}\equiv\sum_{i=1}^{n} q_i\;\omega_{i}^A\otimes\rho_{i}^B\otimes \sigma_{i}^C\;.
\ee
Then, the following are equivalent:
\begin{enumerate}
\item There exists a CPTP map $\mE:\; \B(\mH_B)\to\B(\mH_C)$ such that
	\be\label{main3}
	\mE(\rho_i^B)=\sigma_i^C\quad\forall\;i=1,...,n\;.
	\ee
\item For any $\omega_1^A,...,\omega_n^A\in\B(\mH_{A})$:
\begin{align}\label{eq22}
2^{-H_{\min}(A|B)_{\Omega}}\geq d_C\la\phi_+^{AC}|\Omega^{AC}|\phi_+^{AC}\ra\;,
\end{align}
where $|\phi_{+}^{AC}\>$ denotes the maximally entangled state on $\H_A\otimes\H_C$.
\item For any $\omega_1^A,...,\omega_n^A\in\B(\mH_{A})$:
\be\label{hmin2}
H_{\min}(A|B)_{\Omega}\leq H_{\min}(A|C)_{\Omega}\;.
\ee
\end{enumerate}
\end{lemma}

\begin{proof}
	Consider two families of density matrices, $\{\rho_i^B \}_{i=1}^n$ and $\{\sigma_i^C \}_{i=1}^n$. We want to reformulate, in an equivalent way, the condition
	\[
	\exists \textrm{ CPTP } \mE:\mE(\rho_i^B)=\sigma_i^C,\forall i.
	\]
	By introducing a set of self-adjoint operators $\{X_j^C\}$ forming a basis for $\B(\sH_C)$, Eq.~(\ref{main3}) can be written as
	\begin{equation}\label{eq:vectors}
	\exists \textrm{ CPTP } \mE:\Trr{\mE(\rho_i^B)\ X_j^C}=\Trr{\sigma_i^{C}\  X_j^C},\forall i,j.
	\end{equation}
	
	Let us now consider the set of real vectors 
	\[
	\boldsymbol{r}_{\mE}=(r_{ij}):r_{ij}=\Trr{\mE(\rho_i^B)\ X_j^C}
	\]
	obtained by letting $\mE$ vary over all possible CPTP maps from system $B$ to system $C$, while the $\rho_i$'s and the $X_j$'s are kept fixed. It is clear that the set
	\[
	\mathcal{S}=\{\boldsymbol{r}_\mE:\mE\text{ CPTP} \}
	\]
	is a closed and bounded convex set, as it is the image, under a linear map, of the set of CPTP maps from $B$ to $C$ (that is a closed and bounded convex set). By writing $\boldsymbol{s}=(s_{ij})$ when $s_{ij}=\Trr{\sigma_i^{C}\  X_j^C}$, Eq.~(\ref{eq:vectors}) becomes
	\[
	\boldsymbol{s}\in\mathcal{S}\;.
	\]
	
	At this point, we invoke the separation theorem for convex sets (see, e.g., Ref.~\cite{rockafellar}), which in particular implies the following:
	
	\begin{lemma*}\label{theo:sep}
		Let $\mathcal{S}\in\mathbb{R}^n$ be a closed and bounded convex set. The vector
		$y\in\mathbb{R}^n$ belongs to $\mathcal{S}$,
		i.e. $y\in \mathcal{S}$, if and only if, for any vector $k\in\mathbb{R}^n$, $\max_{x\in\mathcal{S}}k\cdot x\ge k\cdot
		y$.
	\end{lemma*}
	Applied to our case, it yields that condition~(\ref{main3}) is equivalent to
	\begin{align*}
	\forall\boldsymbol{\lambda}=(\lambda_{ij})\text{ with }\lambda_{ij}\in\mathbb{R},\quad \max_{\boldsymbol{r}\in\mathcal{S}}\boldsymbol{r}\cdot\boldsymbol{\lambda}\ge \boldsymbol{s}\cdot\boldsymbol{\lambda}\;,
	\end{align*}
	namely,
	\begin{align*}
	&\forall\boldsymbol{\lambda}=(\lambda_{ij})\text{ with }\lambda_{ij}\in\mathbb{R},\\ &\max_{\mE:\textrm{CPTP}}\sum_{ij}\lambda_{ij}\Trr{\mE(\rho_i^B)\ X_j^C}\ge \sum_{ij}\lambda_{ij}\Trr{\sigma_i^{C}\  X_j^C}.
	\end{align*}
	
	Defining self-adjoint operators $Z_i^C=\sum_{j}\lambda_{ij}X_j^C$, we can reformulate the statement as follows:
	\begin{align*}
	&\forall\text{ self-adjoint } \{Z_i^C\},\\
	&\max_{\mE:\textrm{CPTP}}\sum_{i}\Trr{\mE(\rho_i^B)\ Z_i^C}\ge \sum_{i}\Trr{\sigma_i^{C}\ Z_i^C}.
	\end{align*}	
	Since the condition $\Trr{\rho_i^B}=\Trr{\mE(\rho_i^B)}=\Trr{\sigma_i^C}$ for all $i$ is guaranteed by hypothesis, by writing $\widetilde{Z}_i^C=Z_i^C-z_i \I^C$, for $z_i=d_C^{-1}\Trr{Z_i^C}$, it is clear that the condition can be further reformulated as
	\begin{align*}
	&\forall\textrm{ zero-trace } \{\widetilde{Z}_i^C\},\\
	 &\max_{\mE:\textrm{CPTP}}\sum_{i}\Trr{\mE(\rho_i^B)\ \widetilde{Z}_i^C}\ge \sum_{i}\Trr{\sigma_i^{C}\ \widetilde{Z}_i^C}.
	\end{align*}
	Also, by letting $k=\max_i|\mu_{\min}\widetilde{Z}_i^C|$, we can divide both sides by $kd_C$ so that the condition is reformulated as
	\begin{align*}
	&\forall\textrm{ zero-trace } \{\widetilde{Z}_i^C\} \textrm{ with min eigenvalue }\ge -d_C^{-1},\\ &\max_{\mE:\textrm{CPTP}}\sum_{i}\Trr{\mE(\rho_i^B)\ \widetilde{Z}_i^C}\ge \sum_{i}\Trr{\sigma_i^{C}\ \widetilde{Z}_i^C}.
	\end{align*}
	We then add to both sides the constant $nd_C^{-1}$ (where $n$, we recall, denotes the number of states), obtaining
\begin{align*}
	&\max_{\mE:\textrm{CPTP}}\sum_{i}\Trr{\mE(\rho_i^B)\ (\widetilde{Z}_i^C+d^{-1}_C \I^C)}\\
	&\ge \sum_{i}\Trr{\sigma_i^{C}\ (\widetilde{Z}_i^C+d^{-1}_C \I^C)}.
\end{align*}
	
	At this point, it holds that $\widetilde{Z}_i^C+d_C^{-1} \I^C\ge 0$ and $\Trr{\widetilde{Z}_i^C+d_C^{-1} \I^C}=1$, namely, $\widetilde{Z}_i^C+d_C^{-1} \I^C$ are themselves density matrices $\omega_i^C$. We can therefore reformulate Eq.~(\ref{main3}) as follows:
	\begin{align*}
	&\forall\textrm{ states } \{\omega_i^C\},\\ &\max_{\mE:\textrm{CPTP}}\sum_{i}\Trr{\mE(\rho_i^B)\ \omega_i^C}\ge \sum_{i}\Trr{\sigma_i^{C}\ \omega_i^C}.
	\end{align*}
	
	Let us now arbitrarily fix a probability distribution $q_i$, with $q_{\min}\equiv\min_iq_i>0$.
	Adding to both sides the constant $\frac{1-q_{\min}}{q_{\min}}$, multiplying both sides by $q_{\min}$, and rearranging terms, we obtain
		\begin{align*}
	&\forall\textrm{ states } \{\omega_i^C\},\\ &\max_{\mE:\textrm{CPTP}}\sum_{i}\Trr{\mE(\rho_i^B)\ q_{\min}\left(\omega_i^C+\frac{q_i-q_{\min}}{q_{\min}}\frac{\I^C}{d_C}\right)}\\
	&\ge \sum_{i}\Trr{\sigma_i^{C}\ q_{\min}\left(\omega_i^C+\frac{q_i-q_{\min}}{q_{\min}}\frac{\I^C}{d_C}\right)}.
	\end{align*}
	Noticing that the operators
	\[
q_{\min}\left(\omega_i^C+\frac{q_i-q_{\min}}{q_{\min}}\frac{\I^C}{d_C}\right)
	\]
	are sub-normalized density matrices with trace equal to $q_i$, we can equivalently reformulate condition~(\ref{main3}) as follows:
	\begin{align*}
	&\forall\textrm{ states } \{\omega_i^C\},\\ &\max_{\mE:\textrm{CPTP}}\sum_{i}q_i\Trr{\mE(\rho_i^B)\ \omega_i^C}\ge \sum_{i}q_i\Trr{\sigma_i^{C}\ \omega_i^C},
	\end{align*}
	where $q_i>0$ are arbitrarily fixed probabilities.
	
	The next step is to introduce an auxiliary system $A\cong C$ (i.e., $d_A=d_C$), choose two orthonormal bases $\{|i_A\>\}$ and $\{|i_C\> \}$, and define the maximally entangled state
	\begin{equation}\label{eq:max-ent}
	|\phi^{AC}_+\>\equiv d^{-1/2}_A\sum_{i=1}^{d_A}|i_A\>|i_C\>\;.
	\end{equation}
	Noticing that $\Trr{XY}=d\Trr{X\otimes Y^T\ \phi_+}$, where the superscript $T$ denotes the transposition with respect to the basis in~(\ref{eq:max-ent}), and that $\omega_i$ are density matrices if and only if $(\omega_i)^T$ are, we arrive at
	\begin{align}\label{eq:final-point}
	\forall\textrm{ states }& \{\omega_i^A\},\\ &\max_{\mE:\textrm{CPTP}}\sum_{i}q_i\Trr{\left\{\omega_i^A\otimes\mE(\rho_i^B)\right\}\ \phi_+^{AC}}\nonumber\\
	&\ge \sum_{i}q_i\Trr{\left\{\omega_i^A\otimes\sigma_i^{C}\right\}\ \phi_+^{AC}}.\nonumber
	\end{align}
	
	As shown in Ref.~\cite{Kon09}, the quantity
	\begin{align*}
	&\max_{\mE:\textrm{CPTP}}\sum_{i}q_i\Trr{\left\{\omega_i^A\otimes\mE(\rho_i^B)\right\}\ \phi_+^{AC}}\\
	&=\max_{\mE:\textrm{CPTP}}\<\phi_+^{AC}|(\id\otimes\mE)(\Omega^{AB})|\phi_+^{AC}\>\;,
	\end{align*}
	for $\Omega^{AB}\equiv\sum_iq_i\;\omega^A_i\otimes\rho_i^B$,
	can be written in terms of the conditional min-entropy~(\ref{eq:min-ent-def}) as
	\[
	\frac{1}{d_A}2^{-H_{\min}(A|B)_\Omega}\;,
	\]
	We thus proved that statements~(1) and~(2) of Lemma~\ref{mainlemma} are indeed equivalent.
	
	Moreover, a \textit{sufficient} condition for~(\ref{main3}) is that
	\begin{align*}
	\forall\textrm{ states }& \{\omega_i^A\},\\ &\max_{\mE:\textrm{CPTP}}\sum_{i}q_i\Trr{\left\{\omega_i^A\otimes\mE(\rho_i^B)\right\}\ \phi_+^{AC}}\\
	&\ge \max_{\mF:\textrm{CPTP}} \sum_{i}q_i\Trr{\left\{\omega_i^A\otimes\mF(\sigma_i^{C})\right\}\ \phi_+^{AC}}\;,
	\end{align*}
	namely
	\[
	2^{-H_{\min}(A|B)_\Omega}\ge 2^{-H_{\min}(A|C)_\Omega}\;,
	\]
	where now $\Omega^{AB}$ and $\Omega^{AC}$ are meant as the marginals of the same tripartite extension $\Omega^{ABC}=\sum_iq_i\;\omega_i^A\otimes\rho^B_i\otimes\sigma^C_i$. However, it is easy to verify that the above condition is also necessary: indeed, if~(\ref{main3}) holds, due to the data-processing theorem applied to the conditional min-entropy (see, e.g., Ref.~\cite{Ren05} and~\cite{Toma-book}), $H_{\min}(A|B)_\Omega\le H_{\min}(A|C)_\Omega$. We thus have that statements~(1) and~(3) are also logically equivalent, and hence the proof is complete.
\end{proof}

We are now ready to prove the main theorem.

\textbf{Proof of Theorem~\ref{mainthm}:}
Let $\{Q_k^A\}_{k=1}^{d_{A}^{2}}$ be the dual basis of $\{M_j^A\}$ in $\B(\mH_{A})$,  
that is, $\tr\left[M_j^AQ_k^A\right]=\delta_{jk}$.
Then, since $\{Q_k^A\}$ is itself a basis, we can write
\be
\rho^{AB}=\sum_{k=1}^{d_{A}^{2}}Q_k^A\otimes\widetilde{\rho}^{B}_{k}\quad\text{and}\quad
\sigma^{AC}=\sum_{k=1}^{d_{A}^{2}}Q_k^A\otimes\widetilde{\sigma}^{C}_{k}
\ee
where
\begin{align}
& \widetilde{\rho}^{B}_{j}\equiv\tr_{A}\left[\left(M_{j}^A\otimes \I^B\right)\rho^{AB}\right]\nonumber\\
& \widetilde{\sigma}^{C}_{j}\equiv\tr_{A}\left[\left(M_{j}^A\otimes \I^C\right)\sigma^{AC}\right]
\end{align}
are sub-normalized quantum states (i.e. positive semi-definite matrices). Moreover, since $\rho^A=\sigma^A$
we have $\tr\left[\widetilde{\rho}^{B}_{j}\right]=\tr\left[\widetilde{\sigma}^{C}_{j}\right]\equiv p_j$. 
We therefore conclude that that there exists CPTP map $\mE$ that satisfies~\eqref{bdef} if and only if 
there exists a CPTP map $\mE$ that satisfies
\be\label{sameold}
\sigma^{C}_{j}=\mE\left(\rho^{B}_{j}\right)
\ee
where $\rho^{B}_{j}\equiv\widetilde{\rho}^{B}_{j}/p_{j}$ and $\sigma^{C}_{j}\equiv\widetilde{\sigma}^{C}_{j}/p_{j}$.
To apply Lemma~\ref{mainlemma}, we introduce a system $A'$ with $d_{A'}=d_C$, we fix an arbitrary probability distribution $q_i>0$, and define
\begin{align*}
&\Omega^{A'BC}\equiv\sum_{j=1}^{d_{A}^{2}} q_j\; \omega_{j}^{A'}\otimes\rho_{j}^{B}\otimes\sigma_{j}^{C}=\sum_{j=1}^{d_{A}^{2}} \frac{q_j}{p_{j}^{2}}\times\nonumber\\
&\omega_{j}^{A'}\otimes\tr_{A}\left[\left(M_{j}^A\otimes I^B\right)\rho^{AB}\right]\otimes\tr_{A}\left[\left(M_{j}^A\otimes I^C\right)\sigma^{AC}\right]\;,
\end{align*}
where the states $\omega_i^{A'}$ can vary. Then, taking $q_j=p_j$, we conclude that
	\begin{align}
&\Omega^{A'B}=\Phi\otimes\id\left(\rho^{AB}\right)\;,
\nonumber\\
&\Omega^{A'C}=\Phi\otimes\id\left(\sigma^{AC}\right)\;.\label{eq:from-lemma-to-theo}
\end{align}
Notice that, in case some $p_i=0$, we can redefine the measurement operators $M^A_j\to M^A_j+\delta\I^A$, in such a way that they still span the set $\B(\H_A)$ but have non-zero probability everywhere. With Eq.~(\ref{eq:from-lemma-to-theo}) at hand, the proof of Theorem~\ref{mainthm} follows now from Lemma~\ref{mainlemma}.\qed

\section{Efficiency of Quantum Majorization}

We will show now that the problem of whether there exists a CPTP map $\mE$ such that $\mE(\rho_i)=\sigma_i$ (see Lemma~\ref{mainlemma} above) can be formulated as a semidefinite programming. Following similar lines, also all the other versions of quantum majorization discussed in this paper can be shown to be equivalent to a semidefinite programming. 

We start by noting that~\eqref{eq22} can be written as 
\be
2^{-H_{\min}(A|B)_\Omega}\geq \sum_{i=1}^{n}q_i \tr\left(\omega_i\sigma_{i}^{T}\right)
\ee
where
\be
\Omega^{AB}=\sum_{i=1}^{n}q_i\; \omega_{i}\otimes\rho_{i}\;.
\ee
In the following we absorb the $q_i$s into $\omega_i$s, so that the $\omega_i$s become subnormalized,
satisfying $\sum_{i=1}^{n}\tr[\omega_i]=1$. We get that the above condition is equivalent to the condition $\alpha(t)\geq 1$ for all $t$, where
\begin{align}
\alpha(t)\equiv& \frac{1}{t}\min \tr[\tau]\nonumber\\
& \text{subject to }\;\;I^{A}\otimes \tau\geq \sum_{i=1}^{n}\omega_{i}\otimes\rho_{i},\nonumber\\
			&\quad\quad\quad\quad\quad\sum_{i=1}^{n}\tr\left(\sigma_{i}^T\omega_i\right)=t\;;\;\sum_{i=1}^{n}\tr[\omega_i]=1\;,
\end{align}
with $\omega_i\geq 0$.
After rescaling $\tau'\equiv \frac{1}{t}\tau$ and $\omega_i'\equiv\frac{1}{t}\omega_i$ we get
\begin{align}
\alpha(t)\equiv& \min \tr[\tau']\nonumber\\
& \text{subject to }\;\;I^{A}\otimes \tau'\geq \sum_{i=1}^{n}\omega_{i}'\otimes\rho_{i},\nonumber\\
			&\quad\quad\quad\quad\quad\sum_{i=1}^{n}\tr\left(\sigma_{i}^T\omega_{i}'\right)=1\;;\;\sum_{i=1}^{n}\tr[\omega_i']=1/t\;,
\end{align}
The condition $\alpha(t)\geq 1$ for all $t$ is therefore equivalent to one condition, $\alpha\geq 1$ (more precisely, $\alpha=1$ since it can be shown that $\alpha$ can never exceed 1), where
\begin{align}
\alpha\equiv& \min \tr[Z]\nonumber\\
& \text{subject to }\;\;I^{A}\otimes Z\geq \sum_{i=1}^{n}X_{i}\otimes\rho_{i},\nonumber\\
			&\quad\quad\quad\quad\quad\sum_{i=1}^{n}\tr\left(\sigma_{i}^TX_{i}\right)=1\;;\;X_i\geq 0\;.
\end{align}

We now show that the above minimization problem is an SDP. To see it, 
we define the following vector space, which is a direct sum of $n+2$ Hilbert spaces:
\be
V_1\equiv\mB(\mH^A\otimes\mH^{B})\oplus\mB(\mH^{B})\oplus\mB(\mH^{A})\oplus\cdots\oplus\mB(\mH^A)\;.
\ee
The vector space $V_1$ is consisting of matrices $\zeta\in V_1$ of the form:
\be\label{form}
\zeta=(\eta,Z, X_1,..., X_n)
\ee
where $\eta\in\mB(\mH^A\otimes\mH^{B})$, $Z\in\mB(\mH^{B})$, and $X_i\in \mB(\mH^{A})$ for each $i=1,...,n$.
In addition, we define the vector space $V_2\equiv \mB(\mH^{A}\otimes\mH^B)$, and a linear transformation $\Gamma: V_1\to V_2$ given by:
\be
\Gamma\left(\zeta\right)=I\otimes Z-\sum_{i=1}^{n}X_i\otimes\rho_{i}-\eta
\ee
Clearly, the map above is linear. Set $\sigma\equiv (\boldsymbol{0},\boldsymbol{0},\sigma_{1}^{T},...,\sigma_{n}^{T})$ so that $\tr[\sigma\zeta]=\sum_{i=1}^{n}\tr\left(\sigma_{i}^{T}X_i\right)$. 
We also denote $C\equiv (\boldsymbol{0},I,\boldsymbol{0},...,\boldsymbol{0})$.
 With these notations:
$$
\alpha=\min\left\{\tr[C\zeta]\;\Big|\;\zeta\geq 0\;;\;\Gamma(\zeta)= 0\;,\;\tr[\sigma\zeta]=1\right\}
$$
To bring the above optimization problem to a canonical SDP form, we denote $H_j\equiv\Gamma^{*}(E_j)$, where $j=1,...,d_{A}^{2}d_{B}^{2}$ and $E_j$ is a basis of $\mB(\mH^A\otimes\mH^{B})$. We also denote $H_0\equiv\sigma$. With this notations we get
\begin{align}
\alpha\equiv &\min\tr[C\zeta]\nonumber\\
& \text{ subject to }\;\;\zeta\geq 0\nonumber\\
&\quad\quad\quad\quad\quad\;\;\tr\left(\zeta H_{j}\right)=\delta_{0,j}\quad j=0,1,...,d_{A}^{2}d_{B}^{2}
\end{align}
 
It is interesting to note that the dual problem is given by
\begin{align}
\beta\equiv &\max y\nonumber\\
& \text{ subject to }\;\;y\sigma+\Gamma^{*}(\tau^{AB})\leq C\;\;,\tau^{AB}\in V_2.
\end{align}
where
the dual map $\Gamma^{*}$ is given by:
\begin{align}
& \Gamma^{*}(\tau^{AB})=\nonumber\\
&\left(-\tau^{AB},\tau^{B},-\tr_B\left[\tau^{AB}(I\otimes\rho_{1})\right],...,-\tr_B\left[\tau^{AB}(I\otimes\rho_{n})\right]\right)
\end{align}
Therefore, the dual problem can be expressed as
\begin{align}
\beta=&\max\;y\nonumber\\
&\text{ subject to }\;\;\tau^{AB}\geq 0\;;\;\tau^{B}\leq I\;;\;\text{and }\forall\;i=1,...,n\nonumber\\
&\quad\quad\quad\quad\quad\;\;\;y\sigma_{i}^{T}\leq \tr_B\left[\tau^{AB}(I\otimes\rho_{i})\right]
\end{align}
Note that $\alpha$ (or $\beta$) can be commuted efficiently using standard SDP algorithms if we do not require exact solution and allow for a small error.

\section{The Classical Case: Re-derivation of Thermo-Majorization}

 Thermo-majorization generalizes ordinary majorization in a natural way \cite{black1,sherman,stein,blackwell,torgersen-1970,ruch-schranner,torgersen}. Given two probability distributions $\p=(p_i)$ and $\q=(q_i)$ together with the Gibbs distribution $\boldsymbol{\gamma} = (\gamma_i) = (\frac{1}{Z} e^{-\beta E_i})$ at temperature $T=(k\beta)^{-1}$, we say that $\p$ thermo-majorizes $\q$ and write $\p \succ_T \q$ exactly when the following holds
\begin{equation}\label{au}
\sum_k |p_k - t \gamma_k| \ge \sum_k |q_k -t \gamma_k|,
\end{equation}
for all $t \ge 0$. This can be shown to be equivalent \cite{black1,sherman,stein,blackwell,ruch-schranner,AU1980} to the existence of a stochastic map $S$ such that $S\p = \q$ and $S \boldsymbol{\gamma} = \boldsymbol{\gamma}$. In what follows, we show that quantum majorization reduces to Thermo-majorization in the classical case. In particular, we will show that the conditions in Theorem~\ref{mainthm} (specifically, Eq.~\eqref{hmin2} of Lemma~\ref{mainlemma}) reduces to~\eqref{au}. 
We first start with the semi-classical case.

\subsection{The semi-classical case}

In this case, we assume that the $n$ states, $\{\sigma_{i}^{C}\}$, in Lemma~\ref{mainlemma} commute with each other. Therefore, we can assume that they are all diagonal with respect to a fixed basis. We show now that this immediately implies that the $n$ states $\{\omega_i\}$ in Lemma~\ref{mainlemma} can also be taken to be diagonal in the same basis. In fact, in the following lemma we show that if $\{\sigma_{i}^{C}\}$ are all symmetric with respect to some group, then the states $\{\omega_i\}$ also have the same symmetry.

\begin{lemma}
Using the same notations as in Lemma~\ref{mainlemma},
let $\Delta:\mB(\mH^C)\to\mB(\mH^C)$ be a CPTP map, and suppose $\Delta(\sigma^{T}_{i})=\sigma^{T}_{i}$ for all $i=1,...,n$. Then, in all the statements of Lemma~\ref{mainlemma} we can replace the set 
$\{\omega_{i}^{A}\}$ with the set $\Delta^{\dag}(\omega^{A}_{i})$.
\end{lemma}

\begin{remark}
The lemma above is particularly interesting if the map $\Delta$ corresponds to some symmetry. That is, suppose the states $\{\sigma_{i}^{T}\}$ satisfy $U_g\sigma_{i}^{T}U_{g}^{\dag}=\sigma_{i}^{T}$ for any $g\in G$, where $\{U_g\}$ is some unitary representation of a compact group $G$. In this case, one can take $\Delta$ to be the G-twirling, and thereby assume that all the $\omega_{i}^A$s of Lemma~\ref{mainlemma} are also symmetric with respect to the same representation of $G$. 
\end{remark}

\begin{proof}
The proof follows from the two sides of Eq.~\eqref{eq22}. On one hand,
\begin{align}
\label{b2}
d_C\la\phi_+^{AC}|\Omega^{AC}|\phi_+^{AC}\ra&=\sum_{i}q_i\Trr{\sigma_{i}^{T} \omega_i}\nonumber\\
&=\sum_{i}q_i\Trr{\Delta(\sigma_{i}^{T}) \omega_i}\nonumber\\
&=\sum_{i}q_i\Trr{\sigma_{i}^{T} \Delta^{\dag}(\omega_i)}\;,
\end{align}
where $\Delta^\dag$ is the dual (adjoint) unital map of $\Delta$. On the other hand, if 
\be
I\otimes\tau\geq \sum_{i=1}^{n} q_i\;\omega_{i}^A\otimes\rho_{i}^B
\ee
for some non-normalized state $\tau$, then since $\Delta^{\dag}$ is a unital CP map we get
\be
I\otimes\tau\geq \sum_{i=1}^{n} q_i\;\Delta^{\dag}(\omega_{i}^A)\otimes\rho_{i}^B\;.
\ee
That is,
\be\label{b5}
2^{-H_{\min}(A|B)_\Omega}\ge 2^{-H_{\min}(A|B)_{\Delta^{\dag}\otimes\id(\Omega)}}
\ee
Combining~\eqref{b2} and~\eqref{b5} with~\eqref{eq22} we conclude that if~\eqref{eq22} holds for all states of the form 
$\{\Delta^{\dag}(\omega_{i}^A)\}$ then it holds for any set  of $n$ states $\{\omega_{i}^A\}$.
\end{proof}

The case that we are interested here is the one in which all the $\sigma_i$s are diagonal with respect to some fixed basis. This is the case considered in Corollary~1 of Ref.~\cite{Buscemi}. In this case, we can take $\Delta$ to be the completely decohering map with respect to the fix basis. Since the set$\{\Delta(\omega_i)\}$ consists of diagonal matrices, we can assume w.l.o.g. that all the $\omega_i$s in Lemma~\ref{mainlemma} are diagonal. We can therefore write 
\be
\omega^{A}_{i}\equiv \sum_{x=1}^{d_A}r_{x|i}|x\lr x|
\ee
so that
\be
\Omega^{AB}=\sum_{x=1}^{d_A}|x\lr x|\otimes\sum_{i=1}^{n}q_i r_{x|i}\rho_{i}^{B}
\ee
is a classical quantum state. It is well known that for classical quantum states, the conditional min-entropy can be expressed in terms of a guessing probability~\cite{Kon09}.
In the case that $d_A=2$ the conditional-min entropy of $\Omega^{AB}$ can be further simplified and we get
\begin{align}
&2^{-H_{\min}(A|B)_\Omega}\nonumber\\
&=\min_{\tau}\left\{\tr[\tau]:\;\tau\geq \sum_{i=1}^{n}q_i r_{x|i}\rho_{i}^{B}\quad\forall x=1,2\right\}\nonumber\\
& =\frac{1}{2}+\frac{1}{2}\left\|\sum_{i=1}^{n}q_i (r_{1|i}-r_{2|i})\rho_{i}^{B}\right\|_{1}\;.
\end{align}
However, even if the $\sigma_i$s all commute, it is not enough in general to restrict the comparison only to two-dimensional auxiliary states $\omega_i$, if the goal is that of showing the existence of a CPTP map achieving $\rho_i\to\sigma_i$. If such a restriction is made, what one can show is the existence of a weaker map, namely, a 2-statistical morphism~\cite{FBCMP,Buscemi}, but counterexamples have been shown for which neither a CPTP nor a PTP map exists~\cite{Matsu-counter}.

There are two very important exceptions to this. The first is the case in which there are only two commuting states $\{\rho_1,\rho_2\}$ and two commuting states $\{\sigma_1,\sigma_2 \}$, namely, the case of two classical dichotomies. In this case, already Blackwell showed that two-dimensional commuting states $\omega_i$ suffice~\cite{blackwell}.

The second exception is that of two pairs of qubit density matrices $\{\rho_1,\rho_2\}$ and $\{\sigma_1,\sigma_2 \}$: even if these do not commute, again, two-dimensional commuting states $\omega_i$ suffice~\cite{AU1980}.

\subsection{Thermo-majorization}

In the completely classical case, in addition to the $\omega_i$s, also the set $\{\rho_{i}^{B}\}$ consists of diagonal matrices. Denoting 
\be
\rho_{i}^{B}\equiv \sum_{y=1}^{d_B}s_{y|i}|y\lr y|
\ee
we get that 
\be
\Omega^{AB}=\sum_{x=1}^{d_A} p_{xy} |x\lr x| \otimes |y\lr y|\;\;;\;\;
p_{xy}\equiv\sum_{i=1}^{n}q_i r_{x|i}s_{y|i}\;.
\ee
Now, in this case, the conditional min-entropy is given by
\begin{align}
&2^{-H_{\min}(A|B)_\Omega}\nonumber\\
&=\min_{\tau}\left\{\tr[\tau]:\;I^A\otimes\tau\geq \sum_{x,y} p_{xy} |x\lr x| \otimes |y\lr y|\right\}\nonumber\\
& = \sum_{y}\max_{x} p_{xy}= \sum_{y}\max_{x}\; \r_x\cdot \s_y
\end{align}
where for each $x$ and $y$, $\r_x$ is the $n$-dimensional vector whose components are $\{q_ir_{x|i}\}_{i=1}^{n}$, and $\s_y$ is the $n$-dimensional probability vector whose components are $\{s_{y|i}\}_{i=1}^{n}$. Similarly, denoting by 
\be
\sigma_{i}^{C}\equiv \sum_{z=1}^{d_B}t_{z|i}|z\lr z|\;,
\ee
we conclude that 
\be
2^{-H_{\min}(A|C)_\Omega}=\sum_{z}\max_{x}\; \r_x\cdot \t_z\;,
\ee
 where $\t_z$ is the probability vector whose components are $t_{z|i}$. Therefore, in the classical case, the condition in~\eqref{hmin2} is equivalent to
\be\label{fmain}
\sum_{y}f(\s_y)\geq \sum_{z}f(\t_z)
\ee
for any sub-linear functional $f$ of the form $f(\s)=\max_{x}\; \r_x\cdot \s$. Note that $\sum_y\s_y=\sum_z\t_z=(1,1,...,1)^T$.

Finally, to obtain themo-majorization, we consider the case $n=2$. That is, we have two input states $\rho_1$ and $\rho_2$, and two output states $\sigma_1$ and $\sigma_2$.  We can think of $\rho_2$ and $\sigma_2$ as Gibbs states.
Note that all the vectors $\r_x$, $\s_y$, and $\t_z$ are two-dimensional since $n=2$. Therefore, in this case, it is sufficient to consider in~\eqref{fmain} only sub-linear functionals with two elements; that is, of the form $f(\s)=\max\{ \r_1\cdot \s,\; \r_2\cdot \s\}$ (see~\cite{Dahl1999} for more details). We therefore conclude that  the condition in~\eqref{hmin2} is equivalent to
\be
\sum_{y}\max\{ \r_1\cdot \s_y,\; \r_2\cdot \s_y\}\geq \sum_z \max\{ \r_1\cdot \t_z,\; \r_2\cdot \t_z\}
\ee
for all $\r_1,\r_2\in\mbb{R}^{2}_{+}$.
Using the relation $\max\{a,b\}=\frac{a+b}{2}+\frac{|a-b|}{2}$ for any two real numbers $a$ and $b$, the equation above becomes equivalent to
\be
\sum_{y}| (\r_1-\r_2)\cdot \s_y|\geq \sum_z | (\r_1-\r_2)\cdot \t_z|
\ee
where we used the fact that $\sum_y\s_z=\sum_z\t_z=(1,1,...,1)^T$. Denoting by $\r_1-\r_2\equiv\begin{pmatrix} a \\ b\end{pmatrix}\in\mbb{R}^{2}$, the above equation becomes
\be
\sum_{y}|a s_{y|1}+bs_{y|2}|\geq \sum_z | a t_{z|1} +bt_{z|2}|
\ee
Dividing by $a$ and denoting $r\equiv -b/a$ we conclude that our condition in~\eqref{hmin2} reduces in the classical case to the thermo-majorization condition:
\be
\sum_{y}|s_{y|1}-rs_{y|2}|\geq \sum_z |  t_{z|1} -rt_{z|2}|\quad\forall r\geq 0\;,
\ee
Note that there is an equality above if $r<0$ so we assume w.l.o.g. that $r\geq 0$.

\subsection{Proof of Corollary \ref{thermo-major}}

The proof of Corollary \ref{thermo-major} can now be established.
Suppose we are interested in the conversion of $\rho^A$ into $\sigma^{A'}$ under TPs. Moreover suppose that $[\rho^A, \H^A] = 0$, as explained in the main text one may restrict without loss of generality to $\eta_1$ and $\eta_2$ being incoherent in energy. Therefore the state $\Omega^{RA}$ is a classical state. Since TPs are covariant, and $\rho^A$ is incoherent in energy it implies that the states accessible under this class must also be incoherent in energy and so $[\sigma^{A'}, H^{A'}]=0$ is a necessary condition. Since both input and output states are incoherent the problem reduces to the interconversion of the distributions over energy under stochastic maps that preserve the Gibbs state. This coincides with the conditions for thermo-majorization as stated above.

On the other hand, suppose $[\sigma^{A'}, H^{A'}]= 0$. Now if there exists a a TP map $\E$ such that $\E(\rho^A) = \sigma^{A'}$ it is readily seen that $U'(t) \E(\rho^A)U'(t)^\dagger = \E(U'(t) \rho^AU'(t)^\dagger )=\sigma^{A'}$ for any $t$. Averaging over $t$ gives that $\E(\<\rho^A\>) = \sigma^{A'}$. Therefore $\rho^A \longrightarrow \sigma^{A'}$ under TPs if and only if $\<\rho^A\> \longrightarrow \sigma^{A'}$ under TPs. Therefore such an interconversion is possible if and only if the distribution over energy of $\<\rho^A\>$ thermo-majorizes the distribution over energy of $\sigma^{A'}$.

\section{$G$-Covariant maps}

Theorem~\ref{mainthm} can also be specialized to $G$-covariant maps. In what follows, we consider three unitary representations $g\to U_g$ of the same compact group $G$ on systems $A$, $B$, and $C$. We use the following notations: $\mU_g(x)=U_gxU_g^\dag$, $\overline{\mU}_g(x)=U_g^*xU_g^T$, $\mU^T_g(x)=U_g^TxU_g^*$, and $\mU^\dagger_g(x)=U_g^\dag x U_g$, with obvious meaning of symbols. We also introduce the bipartite twirling operation
\[
\G(x)=\int_Gdg\;\overline{\mU}_g\otimes\mU_g(x)
\]

\subsection{$G$-covariant version of Lemma~\ref{mainlemma}}

\begin{lemma}\label{cov-lemma}
	Let $\{\rho_i^B\}_{i=1}^{n}$ and $\{\sigma_{i}^C\}_{i=1}^{n}$ be two sets of $n$ density matrices in $\B(\mH_B)$ and $\B(\mH_C)$, respectively. Let $\{q_i\}_{i=1}^{n}$ be some arbitrary but fixed probability distribution with $q_i>0$.
	For any set of $n$ density matrices $\{\omega_{i}^A\}_{i=1}^{n}$ in $\B(\mH_{A})$ (with $d_A=d_C$) define the following tripartite separable matrix:
	\[
	\Omega^{ABC}\equiv\sum_{i=1}^{n} q_i\;\omega_{i}^A\otimes\rho_{i}^B\otimes \sigma_{i}^C\;,
	\]
	and its twirled version
	\[
	\widetilde{\Omega}^{ABC}=\int_Gdg\sum_{i=1}^{n} q_i\;\overline{\mU}_g(\omega_{i}^A)\otimes\mU_g(\rho_{i}^B)\otimes \mU_g(\sigma_{i}^C)\;.
	\]
	Then, the following are equivalent:
	\begin{enumerate}
		\item There exists a covariant CPTP map $\mE:\; \B(\mH_B)\to\B(\mH_C)$ such that
		\begin{align*}
		\mE(\rho_i^B)=\sigma_i^C\quad\forall\;i=1,...,n\;.
		\end{align*}
		\item For any $\omega_1^A,...,\omega_n^A\in\B(\mH_{A})$:
		\begin{align*}
		2^{-H_{\min}(A|B)_{\widetilde{\Omega}}}\ge d_C\la\phi_+^{AC}|\widetilde{\Omega}^{AC}|\phi_+^{AC}\ra\;.
		\end{align*}
		\item For any $\omega_1^A,...,\omega_n^A\in\B(\mH_{A})$:
\begin{align*}
		H_{\min}(A|B)_{\widetilde{\Omega}}\leq H_{\min}(A|C)_{\widetilde{\Omega}}\;.
\end{align*}
	\end{enumerate}
\end{lemma}

\begin{proof}
	The proof of Lemma~\ref{mainlemma} goes through unchanged, with the only difference being that we want to find a CPTP map $\mE$ that is covariant, i.e., that satisfies the following property:
	\begin{align}\label{eq:cov-prop}
	\U_g^C [\mE(\rho^B )]=\mE(\U_g^B[\rho^B])\quad\forall g\in G\;.
	\end{align}
	Hence, we can start from Eq.~(\ref{eq:final-point}), which in the covariant case becomes	
	\begin{align*}
	&\forall\textrm{ states } \{\omega_i^A\},\\
	&\max_{\mE:\textrm{ covar. CPTP}}\sum_{i}q_i\Trr{\left\{\omega_i^A\otimes\mE(\rho_i^B)\right\}\ \phi_+^{AC}}\\
	&\ge \sum_{i}q_i\Trr{\left\{\omega_i^A\otimes\sigma_i^{C}\right\}\ \phi_+^{AC}}.
	\end{align*}
	
	Using the covariance of the channel Eq.~(\ref{eq:cov-prop}), and the so-called ``ricochet property'' of the maximally entangled state, that is, $(\I^A\otimes X_C)|\phi^{AC}_+\>=(X^T_A\otimes \I^C)|\phi^{AC}_+\>)$, we can rewrite the left-hand side of the above inequality as follows:
	\begin{align*}
	&\sum_{i}q_i\Trr{\left\{\omega_i^A\otimes\mE(\rho_i^B)\right\}\ \phi_+^{AC}}\\
	&=\sum_{i}q_i\int_Gdg\Trr{(\omega_i^A\otimes\mE(\rho_i^B))\ (\mU_g^T\otimes\mU_g^\dagger)(\phi_+^{AC})}\\
	&=\sum_{i}q_i\int_Gdg\Trr{(\overline{\mU}_g\otimes\mU_g)(\omega_i^A\otimes\mE(\rho_i^B))\ \phi_+^{AC}}\\
	&=\<\phi_+^{AC}|\;(\id_A\otimes\mE_B)(\widetilde{\Omega}^{AB})\;|\phi_+^{AC}\ra\;,
	\end{align*}
	where, we recall, the channel $\mE$ is assumed to be covariant.
	
	Le us now consider the quantity
	\[
	\max_{\mE:\textrm{CPTP}}\<\phi_+^{AC}|(\id_A\otimes\mE_B)(\widetilde{\Omega}^{AB})|\phi_+^{AC}\ra\;,
	\]
	where the maximization now is allowed to run over all possible CPTP maps, not only covariant ones. However, since both $\widetilde{\Omega}^{AB}$ and $\phi_+^{AC}$ are invariant for the action $\overline{\mU}_g\otimes \mU_g$, we immediately have that
	\begin{align*}
	&\max_{\mE:\textrm{CPTP}}\<\phi_+^{AC}|(\id_A\otimes\mE_B)(\widetilde{\Omega}^{AB})|\phi_+^{AC}\ra\\
	&=\max_{\mE:\textrm{CPTP}}\Trr{(\id_A\otimes\mE_B)(\widetilde{\Omega}^{AB})\ \phi_+^{AC}}\\
	&=\int_Gdg\Trr{(\id_A\otimes\mE_B)\circ(\overline{\mU}_g\otimes \mU_g)(\widetilde{\Omega}^{AB})\ (\overline{\mU}_g\otimes \mU_g)(\phi_+^{AC})}\\
	&=\int_Gdg\Trr{(\id_A\otimes\mU_g^\dagger\circ\mE_B\circ\mU_g)(\widetilde{\Omega}^{AB})\ \phi_+^{AC}}\\
	&=\max_{\mE:\textrm{ covar. CPTP}}\<\phi_+^{AC}|(\id_A\otimes\mE_B)(\widetilde{\Omega}^{AB})|\phi_+^{AC}\ra\;,
	\end{align*} 
	and hence, using the conditional min-entropy,
	\begin{align*}
	&\max_{\mE:\textrm{ covar. CPTP}}\<\phi_+^{AC}|(\id_A\otimes\mE_B)(\widetilde{\Omega}^{AB})|\phi_+^{AC}\ra\\
	&=\frac{1}{d_A}2^{-H_{\min}(A|B)_{\widetilde{\Omega}}}\;.
	\end{align*}
	
	Hence, statement~(1) is equivalent to
	\begin{align*}
	&2^{-H_{\min}(A|B)_{\widetilde{\Omega}}}\\
	&\ge d_A\sum_{i}q_i\Trr{\left\{\omega_i^A\otimes\sigma_i^{C}\right\}\ \phi_+^{AC}}\\
	&=d_A\sum_{i}q_i\Trr{\mG\left\{\omega_i^A\otimes\sigma_i^{C}\right\}\ \phi_+^{AC}}\\
	&=d_A\<\phi_+^{AC}|\widetilde{\Omega}^{AC}|\phi_+^{AC}\ra\;.
	\end{align*}
	(Remember that $d_A=d_C$.) Following the same arguments used in the proof of Lemma~\ref{mainlemma}, we also obtain the equivalence between statement~(1) and statement~(3).	
\end{proof}

\subsection{$G$-covariant version of Theorem~\ref{mainthm}}

As before we used Lemma~\ref{mainlemma} to prove Theorem~\ref{mainthm}, here we use Lemma~\ref{cov-lemma} to prove Theorem~\ref{mainthm2}.

\begin{theorem}\label{mainthm2}
	Let $\rho^{AB}\in \B(\H_A \otimes \H_B) $ and $\sigma^{AC}\in \B(\H_A \otimes \H_C)$ be two compatible bipartite quantum states. Denote the dimension of any system $X$ as $d_X\in \mathbb{N}$. The following are equivalent:
\begin{enumerate}
\item There exists a $G$-covariant CPTP map $\mE:\B(\mH_{A})\to\B(\mH_{B})$ such that
\be\label{bdef2}
\sigma^{AC}=\id\otimes\mE\left(\rho^{AB}\right)
\ee
\item For any quantum process (CPTP linear map) $\Phi: \B(\H_A) \to\B(\H_{A'})$, with $d_{A'}=d_C$, 
\begin{align}\label{bhmin-group}
&H_{\min}(A'|B)_{\mG\left[\Phi\otimes\id\left(\rho^{AB}\right)\right]}\\
&\leq H_{\min}(A'|C)_{\mG\left[\Phi\otimes\id\left(\sigma^{AC}\right)\right]}\;.\nonumber
\end{align}
\item Eq.~\eqref{bhmin-group} holds for any measurement-prepare quantum channel 
$\Phi:\B(\mH_{A})\to\B(\mH_{A'})$ of the form:
\be\label{formphi}
\Phi\left(\gamma^A\right)=\sum_{j=1}^{d_{A}^{2}}\tr\left[M_{j}^A\gamma^A\right]\omega_{j}^{A'}\;,
\ee
where $\{M_{j}^A\}$ is an arbitrary, but fixed, informationally complete POVM on system $A$, while the states $\{\omega_{j}^{A'}\}$ can freely vary.
\item For any $\Phi:\B(\mH_{A})\to\B(\mH_{A'})$ of the form~\eqref{formphi} the following holds:
\begin{align*}
&2^{-H_{\min}(A|B)_{\mG\left[\Phi\otimes\id\left(\rho^{AB}\right)\right]}}\\
&\geq d_A\;\la\phi^+|\mG\left[\Phi\otimes\id\left(\sigma^{AC}\right)\right]|\phi^+\ra\;,
\end{align*}
where $|\phi_+^{A'C}\>$ is the maximally entangled state between systems $A'$ and $C$.
\end{enumerate}
\end{theorem}

We are now ready to prove the main theorem.

\textbf{Proof of Theorem~\ref{mainthm2}:}
Let $\{Q_k^A\}_{k=1}^{d_{A}^{2}}$ be the dual basis of $\{M_j^A\}$ in $\B(\mH_{A})$,  
that is, $\tr\left[M_j^AQ_k^A\right]=\delta_{jk}$.
Then, since $\{Q_k^A\}$ is itself a basis, we can write
\be
\rho^{AB}=\sum_{k=1}^{d_{A}^{2}}Q_k^A\otimes\widetilde{\rho}^{B}_{k}\quad\text{and}\quad
\sigma^{AC}=\sum_{k=1}^{d_{A}^{2}}Q_k^A\otimes\widetilde{\sigma}^{C}_{k}
\ee
where
\begin{align}
& \widetilde{\rho}^{B}_{j}\equiv\tr_{A}\left[\left(M_{j}^A\otimes \I^B\right)\rho^{AB}\right]\nonumber\\
& \widetilde{\sigma}^{C}_{j}\equiv\tr_{A}\left[\left(M_{j}^A\otimes \I^C\right)\sigma^{AC}\right]
\end{align}
are sub-normalized quantum states (i.e. positive semi-definite matrices). Moreover, since $\rho^A=\sigma^A$
we have $\tr\left[\widetilde{\rho}^{B}_{j}\right]=\tr\left[\widetilde{\sigma}^{C}_{j}\right]\equiv p_j$. 
We therefore conclude that that there exists a covariant CPTP map $\mE$ that satisfies~\eqref{bdef2} if and only if 
there exists a covariant CPTP map $\mE$ that satisfies
\be
\sigma^{C}_{j}=\mE\left(\rho^{B}_{j}\right)
\ee
where $\rho^{B}_{j}\equiv\widetilde{\rho}^{B}_{j}/p_{j}$ and $\sigma^{C}_{j}\equiv\widetilde{\sigma}^{C}_{j}/p_{j}$.
To apply Lemma~\ref{cov-lemma}, we introduce a system $A'$ with $d_{A'}=d_C$, we fix an arbitrary probability distribution $q_i>0$, and define
\begin{align*}
&\Omega^{A'BC}\equiv\sum_{j=1}^{d_{A}^{2}} q_j\; \omega_{j}^{A'}\otimes\rho_{j}^{B}\otimes\sigma_{j}^{C}=\sum_{j=1}^{d_{A}^{2}} \frac{q_j}{p_{j}^{2}}\times\nonumber\\
&\omega_{j}^{A'}\otimes\tr_{A}\left[\left(M_{j}^A\otimes I^B\right)\rho^{AB}\right]\otimes\tr_{A}\left[\left(M_{j}^A\otimes I^C\right)\sigma^{AC}\right]\;,
\end{align*}
where the states $\omega_i^{A'}$ can vary. The corresponding twirled state is
\[
\widetilde{\Omega}^{A'BC}\equiv\int_Gdg(\overline{\mU}^A_g\otimes\mU^B_g\otimes\mU^C_g)(\Omega^{ABC})\;.
\]

Then, taking $q_j=p_j$, we conclude that
\begin{align}\label{eq:from-lemma2-to-theo2}
&\widetilde{\Omega}^{A'B}=\mG[\Phi\otimes\id\left(\rho^{AB}\right)]\;,
\nonumber\\
&\widetilde{\Omega}^{A'C}=\mG[\Phi\otimes\id\left(\sigma^{AC}\right)]\;.
\end{align}
Notice that, in case some $p_i=0$, we can redefine the measurement operators $M^A_j\to M^A_j+\delta\I^A$, in such a way that they still span the set $\B(\H_A)$ but have non-zero probability everywhere. With Eq.~(\ref{eq:from-lemma2-to-theo2}) at hand, the proof of Theorem~\ref{mainthm2} follows now from Lemma~\ref{cov-lemma}.\qed

\subsection{Covariant Stinespring dilation}

Given systems $A$ and $A'$, with Hilbert spaces $\H_A$ and $\H_{A'}$, we assume that each carry a unitary representation of a compact group $G$ given by $U: G \rightarrow \B(H_A)$ and $U': G \rightarrow \B(\H_{A'})$ respectively. A quantum process $\E : \B(A_A) \rightarrow \B(\H_{A'})$ from $A$ into $A'$ is said to be \emph{covariant} or \emph{symmetric} if $\E \circ \U_g = \U'_g \circ \E$ for all $g\in G$. The following lemma was proved in~\cite{gour2008resource}, and we provide the proof here for convenience.

\begin{lemma} \cite{gour2008resource} Given a covariant quantum process $\E :\B(\H_A) \rightarrow \B(\H_{A'})$ there exists a Kraus decomposition
\begin{equation}
\E(\rho^A) = \sum_{\lambda,m,k} K_{\lambda, m, k} \rho^A K^\dagger_{\lambda, m, k},
\end{equation}
with Kraus operators $K_{\lambda,m,k} : \H_A \rightarrow \H_{A'}$ that transform as 
\begin{equation}\label{ITO-Kraus}
U^{A'}(g) K_{\lambda,m,k} U^A(g)^\dagger= \sum_k v^\lambda(g)_{jk} K_{\lambda,m,j}
\end{equation}
where $(v^\lambda(g)_{jk})$ are the matrix elements of the $\lambda$-irrep of $G$ and $m$ is a multiplicity label.
\end{lemma}
\begin{proof} Let $\{K_i\}$ be a set of linearly independent Kraus operators for $\E$. Since $\E$ is covariant we have that $\U^{A'}_g \circ \E \circ (\U^A_g)^\dagger = \E$ for any $g\in G$, and so it follows that $\{U^{A'}(g)K_i U^A(g)^\dagger\}_i$ forms another set of Kraus operators for $\E$ for any fixed $g\in G$. Since the Kraus representation is unique up to unitary mixing this implies that $U^{A'}(g)K_i U^A(g)^\dagger  = \sum_j V(g)_{ij} K_j$. Moreover, since the Kraus operators are linearly independent it follows that this unitary $V(g)$ is unique for any fixed $g$ and so the matrices $V(g)$ form a non-projective unitary representation of $G$. Using the unitary freedom to choose the basis $\{K_i\}$ we can choose a basis for which $V(g)$ is block diagonal in terms of a sum of unitary irreps of $G$.  We denote this basis $\{K_{\lambda,m, k}\}$, with $\{K_{\lambda,m, k}\}$ transforming as a $\lambda$ irrep under $G$ for each multiplicity $m$ as in Equation (\ref{ITO-Kraus}), and $k$ labels the basis vector of the irrep. This completes the proof.
\end{proof}
Such Kraus operators are said to transform irreducibly under the group action, and are irreducible tensor operators.  
\begin{theorem}[Covariant Stinespring~\cite{marvianthesis}] For any covariant quantum process $\E :\B(\H_A) \rightarrow \B(\H_{A'})$ there exists a Stinespring dilation
\begin{equation}
\E(\rho^A) = \tr_C V \rho^A \otimes |\sigma \>\< \sigma|^B V^\dagger
\end{equation}
where $|\sigma\>^B\in \H_B$ is a symmetric state under the unitary representation $U^B$ of $G$ on system $B$, system $C$ carries a unitary representation $U^C$ of $G$, and $V: \H_A \otimes \H_B \rightarrow \H_{A'} \otimes \H_C$ is an isometry such that
\begin{equation}\label{cov-V}
V (U^A(g) \otimes U^B(g) )= (U^{A'}(g) \otimes U^C(g)) V,
\end{equation}
for all $g \in G$.
\end{theorem}
\begin{remark}
This theorem was proved in~\cite{marvianthesis} for the case $\mH^A=\mH^{A'}$. The proof of the general case is essentially identical, and we provide the proof below for convenience.
\end{remark}
\begin{proof} 
From the previous lemma, a covariant quantum process $\E:\B(\H_A) \rightarrow \B(\H_{A'}) $ always has a Kraus decomposition $\{K_{\lambda,m,k}\}$ such that
\begin{equation}
U^{A'}(g) K_{\lambda,m,k} (U^{A})^\dagger(g)= \sum_k v^\lambda(g)_{jk} K_{\lambda,m,j},
\end{equation}
where $\lambda$ labels an irrep of $G$, $m$ is a multiplicity label and $k$ is the basis vector label of the irrep.

Let $B$ be a system with Hilbert space $\H_B = \rm{span}\{|\sigma\>\}$, with the state $|\sigma\>$ being symmetric under the action of $G$. For any pair $(\lambda,m)$ appearing in the Kraus decomposition of $\E$, let $\mathcal{W}_{(\lambda^*,m)}$ be a Hilbert space isomorphic to the $\lambda^*$-irrep of $G$ and for which we choose a basis $\{|\lambda^*,m,k\>\}_k$. We define $\H_C :=\H_B \bigoplus_{(\lambda,m)} \mathcal{W}_{(\lambda^*,m)}$, where the direct sum ranges over all $(\lambda,m)$ occurring in the Kraus decomposition of $\E$. The space $\H_C$ carries the unitary group action
\begin{equation}
U^C(g) = |\sigma\>\<\sigma| \bigoplus_{(\lambda, m)} \sum_{j,k}(v^{\lambda}(g)_{jk})^* |\lambda^*,m,j\>\<\lambda^*,m,k|,
\end{equation}
where $(v^\lambda(g)_{jk})$ are the unitary matrix components of the irrep $\lambda$ of $G$.

We define the operator $V: \H_A \otimes \H_B \rightarrow \H_{A'} \otimes \H_C$ as
\begin{equation}
V := \sum_{\lambda, m, k} K_{\lambda,m,k} \otimes |\lambda^*, m,k\>\<\sigma|.
\end{equation}
Using that the $\{K_{\lambda,m,k}\}_k$ transform irreducibly under the action of $G$, together with the fact that $(v^\lambda(g)_{jk})$ is a unitary matrix, it is readily verified that Equation (\ref{cov-V}) holds for all $g \in G$, and so $V$ is covariant under the action of $G$. Moreover since $\sum_{\lambda,m,k} K_{\lambda,m,k}^\dagger K_{\lambda,m,k} = \I^A$, and $\{|\lambda^*,m,k\>\}_{\lambda,m,k}$ is an orthonormal set of states we have that $V^\dagger V = \I^A \otimes |\sigma\>\<\sigma|^B$ and so $V$ is an isometry from $\H_A \otimes \H_B$ into $\H_{A'}\otimes \H_C$ . Finally, we have that
\begin{equation}
\E(\rho^A) = \tr_C V (\rho^A \otimes |\sigma \>\< \sigma|^B) V^\dagger,
\end{equation}
and so have constructed the required dilation for the covariant quantum process $\E$.
\end{proof}
The following lemma clarifies that any mixed symmetric state can always be purified to a pure quantum state that is also symmetric under the group action.
\begin{lemma}\label{pure-symmetric} Consider a quantum system $A$, carrying a unitary representation $U^A:G \rightarrow \B(\H^A)$, and a mixed quantum state $\sigma^A$ for which $\U^A_g(\sigma^A) = \sigma^A$ for all $g\in G$. Then, there exists a purification $|\psi^{AB}\>$ of $\sigma^A$ onto a composite system $AB$, and a unitary representation $V^B:G \rightarrow \B(\H^B)$
such that $U^{A}_{g} \otimes V^{B}_{g} |\psi^{AB}\> = |\psi^{AB}\>$ for all $g\in G$.
\end{lemma}
\begin{proof} 
Let $\{|j\ra^A\}_{j=1}^{r}$ be an orthonormal basis of the support subspace of $\sigma^A$, where $r$ is the rank of $\sigma^A$.
Let
\be
|\psi\ra^{AB}=\sigma^{1/2}\otimes I^B|\phi_{+}\ra^{AB}\quad;\quad|\phi_{+}\ra^{AB}\equiv\sum_{j=1}^{r}|j\ra^A|j\ra^B
\ee
be a purification of $\sigma^{A}$. For any $g\in G$, the complex matrix $A_g\equiv U_{g}\sigma^{1/2}$ has a polar decomposition of the form $A_g=P_gV_{g}$, where $P_g=\sqrt{A_gA_{g}^{\dag}}\geq 0$ and $V_{g}^{T}$ is an $r\times r$ unitary matrix. However, since we assume that $\sigma$ is symmetric, we get that for all $g\in G$
\be
P_g=\sqrt{A_gA_{g}^{\dag}}=\sqrt{U_g\sigma U_{g}^{\dag}}=\sigma^{1/2}\;.
\ee
That is, for all $g$
\be
U_{g}\sigma^{1/2}=\sigma^{1/2}V_{g}\;.
\ee
The above relation implies that $V_g$ is a group representation. To see it, note that
\be
U_{gh}\sigma^{1/2}=U_gU_{h}\sigma^{1/2}=U_g\sigma^{1/2}V_{h}=\sigma^{1/2}V_gV_{h}
\ee
and on the other hand, by definition of $V$,
\be
U_{gh}\sigma^{1/2}=\sigma^{1/2}V_{gh}\;.
\ee
Therefore, $V$ is a group representation. Finally, note that
\begin{align}
U_g\otimes\bar{V}_g|\psi\ra^{AB}&=U_g\sigma^{1/2}\otimes\bar{V}_g|\phi_{+}\ra^{AB}\nonumber\\
&=U_g\sigma^{1/2}V_{g}^{\dag}\otimes I^B|\phi_{+}\ra^{AB}\nonumber\\
&=\sigma^{1/2}V_gV_{g}^{\dag}\otimes I^B|\phi_{+}\ra^{AB}=|\psi\ra^{AB}\;,
\end{align}
where we used the property that $I\otimes X|\phi_{+}\ra=X^{T}\otimes I|\phi_{+}\ra$ for any complex matrix $X$. 
Since $V$ is a unitary representation of $G$, so is $\bar{V}$.
This completes the proof.
\end{proof}

\section{Generalized Thermal Processes}

We prove the general result in the presence of thermodynamic observables $\{H^A, X_1^A, \dots, X_n^A\}$, which may have non-trivial commutation relations between them. The case on the Hamiltonian being the only thermodynamic observable follows as a special case of this result.

Assumptions (A1) and (A2), together with the requirement that the resource theory be non-trivial in these observables implies that the free state must take the form of the generalized Gibbs ensemble $\gamma^A$,
\begin{equation}
\gamma^A = \frac{1}{\mathcal{Z}} e^{-\beta (H^A - \sum_k \mu_k X_k^A)}
\end{equation}
for constants $\beta, \mu_1, \dots, \mu_n$. This is picked out by a complete passivity in which one has additional access to an ordered macroscopic `bath' for each observable that can give or take arbitrary amounts of that observable. Given an unbounded number of copies of the free state one wishes to know if one can trivialise the theory in terms of providing an arbitrary displacement for any of these observables. However in the presence of thermodynamic constraints, these are coupled in such a way that one must only consider an ``effective" energy bath with Hamiltonian $\tilde{H} = H - \sum_k \mu_k X_k$. Complete passivity with respect to this observable implies the above generalized Gibbs state through standard arguments.

We now give a precise statement of assumption (A3) in the context of thermodynamic observables $\{H^S, X^S_1, \dots, X^S_n\}$ for any quantum system $S$. We first note there are two components to any TP process $\E$ at the microscopic level: the particular interactions between $A$ and an auxiliary system $B$, and the state $\sigma^B$ of the auxiliary system. Under assumption (A1) there are no couplings present between eigenspaces of different eigenvalues of the additively conserved observables, however this does not mean that coherence cannot be injected into $A$. Assumption (A3) places a minimal constraint on the use of coherence sources outside of $A$. 

\textbf{(A3)  (Incoherence)} If the thermodynamically free process $\E:\B(\H_A) \rightarrow \B(\H_{A'})$ is realised microscopically as
\begin{equation}\label{int-1}
\E(\rho^A) = \tr_C V(\rho^A \otimes \sigma^B) V^\dagger,
\end{equation}
with $V$ obeying Equation (\ref{conservation-law}) then $\E$ is also achievable if we replace $\sigma^B$ with $\G(\sigma^B)$ where
\begin{equation}
\G(\sigma^B) := \int dg U_B(g) \sigma^B U_B(g)^\dagger,
\end{equation}
where $U_B$ is the group representation on $B$ generated by the observables $\{H^B, X_1^B, \dots ,X_n^B\}$, and we interact this state with $A$ through some potentially different isometry $W$ that also obeys (\ref{conservation-law}).
\begin{equation}\label{int-2}
\E(\rho^A) = \tr_C W(\rho^A \otimes \G(\sigma^B)) W^\dagger.
\end{equation}

\begin{lemma} Given a set of thermodynamic observables $\{H^S, X^S_1, \dots, X^S_n\}$ for any quantum system $S$, the set TP of quantum processes from $A$ into $A'$ defined by (A1-A3) coincides with the set GPC of Gibbs-preserving processes on $A$ that are covariant under the group $G$ generated by the thermodynamic observables on $A$ and $A'$.\end{lemma}
\begin{proof} We first show that $TP \subset GPC$. Assumption (A2) ensures that the image of the Gibbs state $\gamma^A$ under $TP$ is the fixed point $\gamma^{A'}$, so it suffices to establish covariance. For any system $S$ we define $X^S_0:=H^S$ so as to make notation compact. Given a process $\E \in TP$, assumption (A1) implies that 
\begin{equation}\label{proof-stines}
\E(\rho^A) = \tr_C V( \rho^A \otimes \sigma^B)V^\dagger,
\end{equation}
for some $V$ that obeys the conservation laws given by Equation (\ref{conservation-law}). In particular, this implies that 
\begin{align}
V\exp \left [ i \sum_{k=0}^n \theta_k X_k^{AB} \right ] &=  \exp \left [ i \sum_{k=0}^n \theta_k X_k^{A'C} \right ] V \\
X_k^{AB}&:=X_k^A \otimes\I^B + \I^A\otimes X_k^B \\
X_k^{A'C}&:=X_k^{A'} \otimes\I^C + \I^{A'}\otimes X_k^C 
\end{align}
for all $k=0, \dots d$ and for all $\theta_k \in \mathbb{R}$. Therefore the observables $\{X_k\}$ generate a representation of a group $G$, with elements $g$ indexed by $(\theta_0, \dots ,\theta_d)$, and $U^{A'C}(g) V = V U^{AB}(g)$ for all $g\in G$. Therefore the process sending any $\chi^{AB} \rightarrow V\chi^{AB} V^\dagger$ is $G$-covariant. As discussed, assumption (A3) says that the above $\sigma^B$ can be taken to be symmetric under this group action: $\U^B_g(\sigma^B) = \sigma^B$. Since discarding systems is $G$-covariant, and also composing of $G$-covariant processes results in a $G$-covariant process, we see that $\rho^A \rightarrow \rho^A\otimes \sigma^B \rightarrow V(\rho^A\otimes \sigma^B)V^\dagger \rightarrow  \tr_C V( \rho^A \otimes \sigma^B)V^\dagger = \E(\rho^A)$ is a $G$-covariant process for any $\E$ of the form (\ref{proof-stines}). Therefore $TP \subset GPC$.

Conversely, let $\E \in GPC$. Since $\E(\gamma^A) = \gamma^{A'}$, assumption (A2) holds automatically. Since $\E$ is $G$-covariant with respect to the group generated by $\{X_k^A\}$ as shown there exists a Stinespring dilation of the process $\E$ of the form
\begin{equation}
\E(\rho^A) = \tr_C V(\rho^A \otimes |\psi\>\<\psi|^B) V^\dagger,
\end{equation}
where $V$ is a $G$-invariant isometry and $|\psi\>^B$ is invariant under the group action on $B$. The invariance of $V$ implies that assumption (A1) holds, while the symmetry of $|\psi\>$ implies that there are no coherences between eigenspaces of the distinguished observables and so (A3) holds. Therefore $\E \in TP$, and so the two sets of processes coincide as claimed.
\end{proof}

\begin{remark}Note that Lemma \ref{pure-symmetric} shows that replacing any auxiliary $\sigma^B$ with its dephased version $\G(\sigma^B)$ as discussed in the main text is consistent with the existence of a Stinespring form in which the auxiliary system is taken to be in a pure symmetric quantum state. Also note that that we implicitly assume that the group $G$ generated by the thermodynamic observables on the input system coincides with the group generated by those on the output system, which is a basic physical requirement.
\end{remark}

If both (A1) and (A3) hold then one can establish the following.
\begin{lemma} If both  Equations (\ref{int-1}) and (\ref{int-2}) hold with $V$ and $W$ respecting the conservation law (\ref{conservation-law}), then (\ref{int-2}) also holds with $W$ replaced with $V$.
\end{lemma}
\begin{proof} We have that
\begin{align}
\tr_C V(\rho^A \otimes \sigma^B) V^\dagger&=\tr_C W(\rho^A \otimes \G(\sigma^B)) W^\dagger
\end{align}
for all $\rho^A$ on $A$. Therefore, for any $g\in G$ we have
\begin{align}
&\hspace{-0.6in}\U_g^\dagger (\tr_C V(\U_g(\rho^A) \otimes \sigma^B) V^\dagger)\nonumber\\
&=\U_g^\dagger(\tr_C W(\U_g (\rho^A) \otimes \G(\sigma^B))) W^\dagger \nonumber \\
&= \E(\rho^A).
\end{align}
but $U_{A'}(g)^\dagger \otimes \I_C V =\I_{A'}\otimes U_C(g) V U_A(g)^\dagger \otimes U^\dagger_B(g)$. And therefore we see that
\begin{align}
\E(\rho^A) &=  (\tr_C V(\rho^A \otimes \U^\dagger_g(\sigma^B)) V^\dagger)
\end{align}
for any $g \in G$. Integrating over all $g\in G$ provides the desired result.
\end{proof}

To summarize, the state interconversion under TPs is equivalent to the following requirement: 
\begin{align}
\E(\rho^A) &= \sigma^{A'} \\
\E(\gamma^A) &= \gamma^{A'}.
\end{align}
where $\E$ is required to be a $G$-covariant process. 
 
\section{Necessary and sufficient conditions for generalized thermal processes} 
 
Using our main result for $G$-covariant interconversion, with $\{\rho_i\} = \{\rho^A, \gamma^A\}$, $\{\omega_i \} = \{\eta^R_1, \eta_2^R\}$ and $\{\sigma_i \} = \{\sigma^A, \gamma^{A'}\}$ the thermodynamic result follows immediately from the general statement, with
\begin{equation}
\Omega^{RA} = \int_G dg U(g) (q \eta^R_1 \otimes \rho^A + (1-q) \eta_1^R\otimes \gamma^A)U(g)^\dagger
\end{equation}
being the relevant bipartite state, and $U(g)$ the group generated by the observables. We also note that the $G$-twirl is defined such that $\int_G dg = 1$, and in the case of time-translations the integral $\<X\>:=\int_G dg X(g)$ is given as $\lim_{T\rightarrow \infty} \frac{1}{T} \int_{-\frac{1}{2} T}^{\frac{1}{2} T} dt X(t) $ for the time-average of any $X(t)$.

 \section{Finite precision and approximate energy incoherence}
 We can replace assumption (A3) with a slightly weaker version that takes into account thatwe only ever experimentally probe to some finite level of precision. The reason this is useful is that it avoids two technicalities: firstly that the time-translation group action is in general non-compact group $\mathbb{R}$, and secondly even if time-translation is the compact $U(1)$ Lie group no finite dimensional representations will exist in which one can encode all group elements into perfectly distinguishable quantum states. We can circumvent both of these technicalities with the following finite precision assumptions.
 
Firstly, we can always approximate any quantum system with one having finite dimension $d <\infty$, for $d$ sufficiently large. Given this finite dimension, any spectrum $\{E_1, E_2, \dots , E_d\}$ for the system's Hamiltonian $H^A$ can be approximated to an arbitrary precision by a set of rational numbers, $\{ \tilde{E}_1 = \frac{a_1}{b_1}, \dots , \tilde{E}_d = \frac{a_d}{b_d}\}$ with $ a_k, b_k \in \mathbb{Z}$ for each $k$ and $\tilde{E}_k$ arbitrarily close to $E_k$.  Thus, for simplicity we assume the Hamiltonian has a spectrum of rational numbers and so the resultant unitary dynamics $U^A(t) = \exp [ -it H^A]$ is periodic for some finite period $\tau< \infty$.
 
The mapping $t \mapsto U^A(t)$ is therefore a unitary representation of the continuous $U(1)$ group on the system $A$. We may further assume that we only ever resolve time intervals $[t_1, t_2]$ with $t_2 - t_1 \ge \epsilon $ for some small yet finite level of precision $\epsilon >0$. More formally this means that we can replace the $U(1)$ group with the discrete $\mathbb{Z}_N$ action, where $N\epsilon =\tau$ and 
\begin{equation}
n \mapsto U^A(n\epsilon) = e^{-i n\epsilon H^A},
\end{equation}
with $n = 0, 1, \dots N-1$. Therefore the dynamics of any single quantum system can always be approximated by such a discrete, finite action for some $N \in \mathbb{N}$ and sufficiently large. 

In the case that we have multiple systems $A_1, A_2, \dots ,A_M$ with periods $\tau_1, \tau_2, \dots, \tau_M$ respectively, we may choose $\tau = \prod_{k=1}^M\tau_k$ as the time-scale for the composite system. Therefore for multiple systems, there will always exist an $N \in \mathbb{N}$, sufficiently large so that the mapping $n \mapsto \exp [ - i n \epsilon H^{A_k}]$ is a unitary representation of $\mathbb{Z}_N$ on each $\H_{A_k}$, and which approximates the unitary dynamics of each $A_k$ under its Hamiltonian to the specified level of precision. Given this, condition (A3) for incoherence of thermal processes can be replaced with the following.

\textbf{($\mathbf{A3'}$) Approximate incoherence}. Consider the case of the Hamiltonian being the only thermodynamic observable, and assume the finite precision approximations described above. If the thermodynamically free process $\E:\B(\H_A) \rightarrow \B(\H_{A'})$ is realised microscopically as
\begin{equation}
\E(\rho^A) = \tr_C V(\rho^A \otimes \sigma^B) V^\dagger,
\end{equation}
with $V$ obeying Equation (\ref{conservation-law}) then we also have
\begin{equation}
\E(\rho^A) = \tr_C W(\rho^A \otimes \G_\epsilon(\sigma^B)) W^\dagger.
\end{equation}
with $\G_\epsilon(\sigma^B)$ being the group average over $\mathbb{Z}_N$ of the state $\sigma^B$ given by
\begin{equation}
\G_\epsilon(\sigma^B) :=\frac{1}{N}\sum_{n=0}^{N-1} U^B_\epsilon(n) \sigma^B U^B_\epsilon(n)^\dagger,
\end{equation}
with $U^B_\epsilon(n) := \exp[ -i n \epsilon H^B]$ is the finite precision time evolution on $B$, and we interact this state with $A$ through some potentially different isometry $W$ that also obeys (\ref{conservation-law}).

This implies that the constraint of time-translation covariance is replaced with $\mathbb{Z}_N$-covariance to this level of precision. Given this, the analysis for state interconversion may be repeated under ($A3'$) and results in the replacement of $\frac{1}{\tau}\int_0^\tau dt (\cdot) $ with $\frac{1}{N} \sum_{k=0}^{N-1} (\cdot)$ and $U^R(t)\otimes U^A(t)$ by the discrete approximation $U^R_\epsilon(n) \otimes U^A_\epsilon (n)$.
 
\section{Clock times and guessing probabilities}
As in the previous section, we may restrict our attention to a fully discrete setting with quantum systems of finite dimension and finite level of precision $\epsilon$ for time resolution. Covariance of the dynamics is now described with respect to the discrete group $\mathbb{Z}_N$ for some sufficiently large $N\in \mathbb{N}$.

For $q\rightarrow1$ we obtain the $\mathbb{Z}_N$ covariance constraint alone, and the corresponding state $\Omega^{RA}$ takes the form
\begin{equation}\label{discrete-omega}
\Omega^{RA} = \frac{1}{N} \sum_{k=0}^{N-1} U^R_\epsilon(n)\eta^R_1 (U^R_\epsilon(n))^\dagger \otimes U^A_\epsilon(n) \rho^A (U^A_\epsilon(n))^\dagger.
\end{equation}
For a sufficiently large reference frame $R$ there exists a Hamiltonian $H^R$ such that $R$ allows a perfect encoding of the group elements of $G=\mathbb{Z}_N$. In particular for $\rm{dim}(\H_R) = N$ with orthonormal basis $\{|E_k\>^R\}$, we can choose
\begin{equation}
U^R_\epsilon(1) = \sum_{k=0}^{N-1} \omega^k | E_k\>\<E_k|^R,
\end{equation}
where $\omega := e^{\frac{2\pi i}{N}}$ is an $N^{\rm th}$ root of unity. We then have that $(U^R_\epsilon(1))^n = U^R_\epsilon(n)$ for any $n=1, 2, \dots$ and $U^R_\epsilon( N) = U^R_\epsilon(0) = \I^R$ as required. 

Defining $|k\>^R := F|E_k\>^R$, with $F$ being the discrete Fourier transform operator
\begin{equation}
F = \frac{1}{\sqrt{N}}\sum_{i,j=0}^{N-1} \omega^{ij} |E_i\>\<E_j|^R ,
\end{equation}
 it is readily seen that 
\begin{equation}
U^R_\epsilon(n) |0\>^R = | n\>^R,
\end{equation}
and $\<n | m\>^R = 0$ for $n \ne m$ and equal to $1$ for $n=m$. Therefore the reference system $R$ provides a perfect classical encoding of the group elements of $\mathbb{Z}_N$ in the pure states $\{|k\>^R\}$.

Setting $\eta_1^R = |0\>\<0|^R$ in equation \ref{discrete-omega} gives the classical-quantum state
\begin{equation}
\Omega^{RA} =  \frac{1}{N} \sum_{k=0}^{N-1} |k\>\<k|^R \otimes  \rho^A (n).
\end{equation}
where we define $\rho^A(n) := U^A_\epsilon(n) \rho^A (U^A_\epsilon(n))^\dagger$ for the state of $A$ at time $t= n \epsilon$. These states fully encode the set of \emph{clock times} $t=0, \epsilon, \dots ,n\epsilon, \dots, (N-1)\epsilon$ for the joint system.

Since $\Omega^{RA}$ is a classical-quantum state, we have that~\cite{Kon09}
\begin{equation}
H_{\rm min} (R|A)_\Omega = -\log p_{\rm guess},
\end{equation}
where $p_{\rm guess}$ is the optimal Helstrom guessing probability for the ensemble of states $\{(\frac{1}{N}, \rho^A(n))\}_{n=0}^{N-1}$ on $A$. This implies that $2^{-H_{\rm min}(R|A)_\Omega}$ is the optimal guessing probability of the clock time $t=n\epsilon$ for the joint system, given the single copy of $\rho^A$. Monotonicity of $H_{\rm min}(R|A)_\Omega$ under the thermal processes implies monotonicity of the clock time guessing probability for the system.

\section{Reformulation of $H_{\rm min}(R|A)_\Omega$ in terms of the $L_2$-norm}
The clock time guessing probability condition applies for the case where the reference $R$ permits a perfect classical encoding of the clock times, and where we choose $\eta_1^R$ to be one of the clock states. It is natural to ask if this interpretation applies if one perturbs around these assumptions. 

To provide partial insight into this, we can exploit the fact that $\Omega^{RA}$ is separable, and reformulate $H_{\rm min}(R|A)_\Omega$ in terms of a minimization involving a norm distance between the orbit of $\rho^A$ under its Hamiltonian and a reference orbit obtained from $\eta_1^R$. Informally, this can be viewed as a synchronisation task. Note that similar norm expressions arise in the theory of equilibration and so tools from that area may be of use for future analysis.

 Indexing $(\rho_1, \rho_2) =(\rho^A, \gamma^A$) and $q_1 =q$, $q_2 = 1-q$, we can express the entropy as
\begin{align}
&2^{-H_{\rm min}(R|A)_\Omega}\\ 
&= d_{A'}\max_{\mE:\textrm{ covar. CPTP}} \tr [ id_R \otimes \E (\Omega^{RA}) \phi^+] \nonumber\\
&= \max_{\mE:\textrm{ covar. CPTP}}\sum_i q_i \int dg \tr [ (\U^*_g( \eta_i))^T\E(\U_g (\rho_i)] \nonumber \\
&= \max_{\mE:\textrm{ covar. CPTP}}\sum_i q_i \int dg \tr [ (\U_g( \eta_i^T))\E(\U_g (\rho_i)] \nonumber \\
&= \max_{\mE:\textrm{ covar. CPTP}}\sum_i q_i \int dg \tr [ ( \eta_i(g))\E( (\rho_i(g))] \nonumber \\
\end{align}
where we define $\eta_i(g) := \U_g(\eta_i^T)$ and $\rho_i(g) := \U_g(\rho_i)$. Note that $\eta_i^T$ is a quantum state if and only if $\eta_i$ is a quantum state, and so we can simply replace $\eta_i^T \rightarrow \eta_i$ without affecting the result. 

We next use  $\|X - Y\|_2^2= \tr((X-Y)^2) = \tr X^2 + \tr Y^2 - 2 \tr XY$ for any Hermitian operators $X,Y$. And noting that $\tr (\eta_i(g)^2) = \tr \U_g(\eta_i^2) = \tr ( \eta_i^2)$, this implies that
\begin{align}
&H_{\rm min}(R|A)_\Omega =1- \log \left [ \sum_i \tr q_i\eta_i^2 +\right .\nonumber\\
&\left . \hspace{-0.4cm}+\hspace{-0.6cm}\max_{\mE:\textrm{ covar. CPTP}} [ \sum_i q_i \tr (\E(\rho_i)^2) - \int dg \sum_i q_i ||\eta_i(g) - \E(\rho_i(g))||_2^2] \right ]
\end{align}
Defining $\mathcal{P}(\boldsymbol{\eta}) := \sum_i q_i \tr (\eta_i^2)$ for the average purity of an ensemble of states  $\{(q_i,\eta_i)\}$ we have
\begin{align}
&H_{\rm min}(R|A)_\Omega =  \nonumber\\
&1 -\log \Big[ \mathcal{P}(\boldsymbol{\eta}) -\hspace{-0.6cm}\min_{\mE:\textrm{ covar. CPTP}}\Big[\nonumber \\
&\ \ \ \ \ \ \ \ \int_G dg \sum_i q_i \|\eta_i(g) - \E(\rho_i(g))\|_2^2-\int_G dg\mathcal{P}(\boldsymbol{\E(\rho(g))})\Big]\Big].
\end{align}
Since the optimal $\E$ is covariant we have that $\int_G dg  \mathcal{P}(\boldsymbol{\E(\rho(g))}) = \mathcal{P}(\boldsymbol{\E(\rho)})$. In the case of only time-translation we obtain
\begin{align}
&H_{\rm min}(R|A)_\Omega = 1-\log [\mathcal{P}(\boldsymbol{\eta}) \nonumber\\
&-\hspace{-0.6cm}\min_{\mE:\textrm{ covar. CPTP}}[\sum_i q_i \<||\eta_i(t) - \E(\rho_i(t))||_2^2\> -\mathcal{P}(\boldsymbol{\E(\rho)})]].
\end{align}
which expresses the conditional entropy in terms of a minimization of ensemble square distance in the 2-norm with an added purity constraint on the output states. However the covariance of $\E$ and the unitary invariance of the norm also imply that $ ||\eta_i(t) - \E(\rho_i(t))||_2 = ||\eta_i - \E(\rho_i)||_2  = \<||\eta_i (t) - \E(\rho_i(t))||_2\>$ and so $H_{\rm min}(R|A)_\Omega$ can be obtained by computing at any fixed external time $t$, which is consistent with the global state $\Omega^{RA}$ being invariant under time-translation.

\end{document}